\documentclass[12pt,reqno]{amsart}
\usepackage{amsthm}
\usepackage{amsmath}
\usepackage{latexsym}
\usepackage{amsfonts}
\usepackage{amssymb}
\usepackage{color}
\usepackage{bbm,dsfont}
\usepackage{graphicx}
\usepackage{hyperref}
\usepackage{enumerate}
\usepackage{comment}

\newtheorem{proposition}{Proposition}

\newtheorem{corollary}{Corollary}

\theoremstyle{definition}

\newtheorem{remark}{Remark}
\newtheorem{example}{Example}
\newtheorem{definition}{Definition}

\newcommand{\blue}[1]{\textcolor{blue}{#1}}

\newcommand{\R}{\mathbb{R}} 
\newcommand{\nat}{\mathbb N} 
\newcommand{\half}{\tfrac{1}{2}} 

\newcommand{\tr}[1]{{\rm tr}\left[#1\right]} 
\newcommand{\id}{\mathbbm{1}} 

\newcommand{\vv}{\mathbf{v}} 
\newcommand{\vsigma}{\boldsymbol{\sigma}} 

\newcommand{\Mo}{\mathsf{M}}

\newcommand{\M}[2]{\mathcal{M}_{#1,#2}} 
\newcommand{\Mr}[2]{\mathcal{M}^{row}_{#1,#2}} 
 
\newcommand{\Mrall}{\mathcal{M}^{row}} 
\newcommand{\C}[3]{\mathcal{C}_{#1,#2}(#3)} 
\newcommand{\Call}[1]{\mathcal{C}(#1)} 
\newcommand{\Csr}[3]{\bar{\mathcal{C}}_{#1,#2}(#3)} 
\newcommand{\uleq}{\preceq} 
\newcommand{\ul}{\prec} 
\newcommand{\ugeq}{\succeq} 

\newcommand{\state}{\mathcal{S}} 
\newcommand{\effect}{\mathcal{E}} 
\newcommand{\obss}{\mathcal{O}(\mathcal{S})} 
\newcommand{\qubit}{\mathcal{Q}_2} 
\newcommand{\qudit}{\mathcal{Q}_d} 
\newcommand{\cld}{\mathcal{S}^{cl}_d} 
\newcommand{\cln}{\mathcal{S}^{cl}_n} 
\newcommand{\bit}{\mathcal{S}^{cl}_2} 

\newcommand{\rank}[1]{\mathrm{rank}(#1)} 
\newcommand{\rankw}{\mathrm{rank}} 
\newcommand{\nrank}[1]{\mathrm{rank}_+(#1)} 
\newcommand{\nrankw}{\mathrm{rank}_+} 
\newcommand{\psdrank}[1]{\mathrm{rank}_{psd}({#1})} 
\newcommand{\psdrankw}{\mathrm{rank}_{psd}} 
\newcommand{\lmax}[1]{\lambda_{max}(#1)} 
\newcommand{\lmaxx}{\lambda_{max}} 

\newcommand{\lminw}{\lambda_{min}} 

\makeatletter
\renewcommand*\env@matrix[1][\arraystretch]{%
  \edef\arraystretch{#1}%
  \hskip -\arraycolsep
  \let\@ifnextchar\new@ifnextchar
  \array{*\c@MaxMatrixCols c}}
\makeatother

\begin{document}\setlength{\arraycolsep}{2pt}

\title[]{Communication tasks in operational theories}

\author[Heinosaari]{Teiko Heinosaari}
\author[Kerppo]{Oskari Kerppo}
\author[Lepp\"aj\"arvi]{Leevi Lepp\"aj\"arvi}

\email{Teiko Heinosaari: teiko.heinosaari@utu.fi}
\email{Oskari Kerppo: \ \ \ oeoker@utu.fi}
\email{Leevi Lepp\"aj\"arvi: leille@utu.fi}

\address{QTF Centre of Excellence, Turku Centre for Quantum Physics, Department of Physics and Astronomy, University of Turku, FI-20014 Turku, Finland}

\begin{abstract}
We investigate the question which communication tasks can be accomplished within a given operational theory. 
The concrete task is to find out which communication matrices have a prepare-and-measure implementation with states and measurement from a given theory. 
To set a general framework for this question we develop the ultraweak matrix majorization in the set of communication matrices. This preorder gives us means to determine when one communication task is more difficult than another. Furthermore, we introduce several monotones which can be used to compare and characterize the communication matrices. 
We observe that not only do the monotones allow us to compare communication matrices, but also their maximal values in a given theory are seen to relate to some physical properties of the theory. 
The maximal values can then be thought as `dimensions', allowing us to compare different theories to each other. We analyse the introduced monotones one by one and demonstrate how the set of implementable communication matrices is different in several theories with the focus being mainly on the difference between classical and quantum theories of a given dimension.
\end{abstract}

\maketitle

\section{Introduction}

There has been recently several studies on quantum prepare-and-measure scenarios from different points of view.
This topic connects to several active research areas, including self-testing (see e.g. \cite{TaKaVeRoBr18,FaKa19}), dimension witnesses (see e.g. \cite{GaBrHaAc10,BoQuBr14,SiVaWe16,Vicente19}) and foundational principles of quantum theory (see e.g. \cite{DaBrToBuVe17,DuAm18}).
In the current work our main aim is identify and investigate the mathematical structure and general features that this kind of question has in any operational theory.

We start by recalling the convex operational formulation of operational theories (also called general probabilistic theories, see e.g. \cite{Lami17} for more details). A system is described by its \emph{state} $s$ which we assume to be an element of a compact convex subset $\state$ of a finite-dimensional real vector space $\mathcal{V}$; we call $\state$ the state space of the theory. The convexity of $\state$ is a result of the possibility to mix states. Given a state space $\state$ we take the set of \emph{effects} $\effect(\state) \subset \mathcal{V}^*$, the simplest types of measurements, to consist of linear funtionals $e: \state \to [0,1]$ giving probabilities on states: $e(s)$ is interpreted as the probability that the event corresponding to the effect $e$ was registered in an experiment when the system was measured in a state $s$. A \emph{measurement} $M$ with $n$ outcomes is then taken to be a collection of effects $M_1, \ldots, M_n$ such that $\sum_{i=1}^n M_i(s) = 1$ for all states $s \in \state$ thus guaranteeing that some outcome is always  registered for all states.
We denote by $\qudit$ and $\cld$ the state spaces of $d$-dimensional quantum and classical systems, respectively.

By a \emph{communication matrix} (also called \emph{channel matrix} in \cite{FrWe15}) we mean any row-stochastic matrix $C$, i.e., a matrix with non-negative entries with each row suming to 1.
The interpretation in the current investigation is as follows.
Alice has a finite collection of states, called a \emph{state ensemble}, which we describe as a map $s:a\mapsto s_a$ from a finite set $\{1,\ldots,n\}$ to $\state$. 
Alice selects a label $a$ and sends a system in the respective state $s_a$ to Bob. Bob makes a measurement using a fixed measurement $M$, having possible outcomes $\{1,\ldots,m\}$.
The collection of all conditional probabilities describing this preparation-measurement scenario is written as a $n \times m$ communication matrix
\begin{equation}\label{eq:implementation}
C_{ab}= M_b(s_a) \, ,
\end{equation} 
and we then say that $C$ is \emph{implemented} with the pair $s,M$.
We remark that this is the simplest type of prepare-and-measure scenario; one may look for statistics coming from multiple measurements and then the object under investigation is often called a \emph{behaviour}. 
In this work we focus solely on communication matrices.

We can consider any communication matrix $C$ as a communication task.
The question then is: if communication is limited to systems of certain type (i.e. $\state$ is fixed), can one implement $C$?
We denote by $\C{n}{m}{\state}$ the set of all $n\times m$ communication matrices that have an implementation of the form \eqref{eq:implementation} for some $s,M$ belonging to the theory determined by $\state$.
We further denote by $\Call{\state}$ the set of $\state$-implementable communication matrices of finite size, i.e., $\Call{\state} = \cup_{n,m} \C{n}{m}{\state}$.

The starting point of the present investigation is to have a theory independent definition of when some communication task is more difficult than another one.
This will be defined as a preorder in the set of all communication matrices.
As we are going to see, in the classical theories there is one implementable communication task, namely the distinguishability of $n$ states, so that all other implementable communication tasks are easier than this task.
This is the unique feature of classical theories; in any non-classical theories there are implementable communication matrices that do not have a common upper bound in the set of implementable communication matrices.
We will then develop monotones for the defined preorder and link them to physical properties of $\state$.

In the literature, it is often assumed that Alice and Bob are sharing a global source of randomness, i.e., they have \emph{shared randomness}. 
With shared randomness Alice and Bob can implement any mixture of implementable communication matrices since they can coordinate the mixing of preparators and measurement devices.
We denote by $\Csr{n}{m}{\state}$ the convex hull of $\C{n}{m}{\state}$ and this set hence corresponds to all those $n \times m$ communication matrices that Alice and Bob can implement if they have additionally a source of shared randomness.
In Section \ref{sec:mix} we demonstrate that $\C{3}{3}{\qubit}$ is not convex, hence shared randomness is a truly additional resource when forming communication matrices.
Interestingly, it has been shown in \cite{FrWe15} that $\Csr{n}{m}{\qudit}=\Csr{n}{m}{\cld}$ for classical and quantum systems with the same operational dimension $d$. 
Therefore, in the considered simple communication tasks (i.e. consisting of single measurement devices) quantum systems provide no benefit over classical systems of the same dimension if shared randomness is available.
The non-convexity of $\C{n}{m}{\qudit}$ gives an additional motivation to identify the relevant mathematical structure of this set and urges to understand what are the quantum advantages when shared randomness is not used.

This paper is organized as follows. 
In Section \ref{sec:some} we provide some concrete examples of communication matrices that have physical interpretations. 
They are then used to illustrate the developments in later sections.
To motivate our approach, the non-convexity of $\C{3}{3}{\qubit}$ is proven in Section \ref{sec:mix}.
The essential concept of our investigation, namely, ultraweak matrix majorization, is defined and explained in Section \ref{sec:uw}.
The main contribution of this paper is to introduce and develop ultraweak monotones; in Section \ref{sec:monotones} we introduce six monotones and in Section \ref{sec:example} demonstrate their usefulness. 
In Section \ref{sec:comparison} we further investigate the introduced monotones.

\section{Some specific communication matrices}\label{sec:some}

We denote by $\M{a}{b}$ the set of $a\times b$ real matrices and by $\Mr{a}{b}$ the set of $a\times b$ row-stochastic matrices, i.e., those matrices that have nonnegative entries and the entries in each row sum to $1$.
Further, we denote by $\Mrall$ the set of all row-stochastic matrices of finite size, i.e., $\Mrall = \cup_{a,b} \Mr{a}{b}$.
In the following we list some subclasses of communication matrices that have specific physical interpretations.
We will later use these matrices to exemplify the general developments.

The possibility to distinguish $n$ states corresponds to the communication matrix $\id_n$; the distinguishability of states $s_1,\ldots,s_n$ means that there is a measurement $M$, with outcomes $1,\ldots,n$, such that
\begin{equation}\label{eq:implementation-delta}
M_b(s_a) = \delta_{ab}
\end{equation}  
for all $a,b=1,\ldots,n$.
In the other extreme, we denote by $V_n$ the $n\times n$ matrix that has $\tfrac{1}{n}$ everywhere.
This communication matrix is useless for all communication tasks.

The previous communication matrices belong to the following family of communication matrices. 
For every $0 \leq \epsilon\leq 1$, we denote by $D_{n,\epsilon}$ the $n\times n$ communication matrix that has $1-\epsilon$ in the diagonal and $\epsilon/(n-1)$ elsewhere, e.g.,
\begin{align*}
D_{3,1/3}= \begin{bmatrix}[1.3]
\frac{2}{3} & \frac{1}{6} & \frac{1}{6} \\ \frac{1}{6} & \frac{2}{3} & \frac{1}{6} \\ \frac{1}{6} & \frac{1}{6} & \frac{2}{3}
\end{bmatrix} \, .
\end{align*}
Clearly, $D_{n,0}=\id_n$ and $D_{n,1-1/n}=V_n$.
For $0<\epsilon<1-1/n$ we interpret the matrix $D_{n,\epsilon}$ as a noisy unbiased distinguishability matrix; it corresponds to the ability to distinguish $n$ states with the error probability $\epsilon$ and so that the probability of getting a wrong outcome is equal for all wrong outcomes. 
For $1-1/n<\epsilon \leq 1$ the off-diagonal elements are larger than the diagonal elements and for this reason a different interpretation is more natural.
Firstly,
\begin{equation}\label{eq:D-n-1}
D_{n,1} = \frac{1}{n-1} \left[\begin{array}{ccccc} 0 & 1 & 1 & \cdots & 1 \\1 & 0 & 1 & \cdots & 1 \\ 1 & 1  & 0 & & 1 \\ \vdots & &   & \ddots  \\ 1 & \cdots & \cdots & 1 & 0 \end{array}\right] \, .
\end{equation}
and this matrix is related to the \emph{uniform antidistinguishability} of $n$ states.
Namely, we recall that the antidistinguishability of states $s_1,\ldots,s_n$ means that for each index $a$ the corresponding outcome never occurs \cite{Leifer14}. 
We can further require that the other outcomes occur with equal probabilities and this then leads to the concept of uniform antidistinguishability \cite{HeKe19}.
It follows that for $1-1/n<\epsilon<1$ we can regard $D_{n,\epsilon}$ as a noisy uniform antidistinguishability matrix.
For later use we note that this class of matrices (with fixed $n$) is a monoid under the matrix multiplication and, in addition, $V_n$ is an absorbing element.
The product of $D_{n,\epsilon}$ and $D_{n,\mu}$ gives
\begin{equation}
D_{n,\epsilon} D_{n,\mu} = D_{n,\epsilon + \mu - \frac{n}{n-1} \epsilon \mu} \, .
\end{equation}

The uniform antidistinguishability can be seen as a special case of a more general task of communication of partial ignorance.
The mathematical formulation of this type of task leads to the following matrices \cite{HeKe19}.
For every integer pair $(n,t)$ with $n\geq 2$, $1\leq t \leq n-1$, we denote by $G_{n,t}$ the ${{n}\choose{t}} \times n$ communication matrix that has the first row 
$$
1/(n-t)\left[\begin{array}{cccccc} 1 & \cdots & 1 & 0 & \cdots 0  \end{array}\right]
$$
with $n-t$ ones and $t$ zeros.
The other rows are all possible permutations of this that give a different row, written in decreasing lexicographical order.
For instance, 
$$
G_{4,2} = \frac{1}{2} \left[\begin{array}{cccc} 1 & 1 & 0 & 0\\ 1 & 0 & 1 & 0\\ 1 & 0 & 0 & 1\\ 0 & 1 & 1 & 0\\ 0 & 1 & 0 & 1\\ 0 & 0 & 1 & 1  \end{array}\right] \, .
$$
As special cases, we have $G_{n,n-1}=\id_n$ and $G_{n,1}=A_n$, where
\begin{equation}\label{eq:anti-n}
A_n = \frac{1}{n-1} \left[\begin{array}{ccccc} 1 & 1 & \cdots & 1 & 0 \\1 & 1 & \cdots & 0 & 1 \\ 1 & 1  & 0 & & 1 \\ \vdots & &   & \ddots  \\ 0 & \cdots & \cdots & 1 & 1 \end{array}\right] \, .
\end{equation}
The matrix $A_n$ has the same rows as $D_{n,1}$ but in a different order.
They are in fact equivalent matrices in the way that will be clarified in Section \ref{sec:uw}.

The matrices $G_{n,t}$ have a natural interpretation as communication tasks where the goal is to communicate which choices should be avoided. 
Suppose that Alice and Bob play the following game with Charlie. In the game Charlie has $n$ boxes and he chooses to hide a prize in one of the boxes. Charlie then reveals $t$ empty boxes to Alice who must communicate this information to Bob, or more precisely, enable Bob to avoid the empty boxes. Alice and Bob win the game if Bob guesses the box with the prize correctly. It was shown in \cite{HeKe19} that in this game the best strategy for Alice and Bob is to implement the communication matrix $G_{n,t}$ when Charlie has $n$ boxes and $t$ boxes are revealed empty in each round.

Naturally, the `classicality' and `quantumness' of a communication matrix make sense only relative to dimension of the classical state space or the Hilbert space dimension; any communication matrix has a classical implementation with a suitably big classical state space.
Therefore, interesting communication matrices are those that have a realization with $d$-dimensional quantum system but not with $d$-dimensional classical system.
For instance, one can verify that $A_n\in\Call{\bit}$ if and only if $n=2$ while it has been shown in \cite{HeKe19} that $A_n\in\Call{\qubit}$ if and only if $n\in\{2,3,4\}$.

\section{Mixtures of communication matrices}\label{sec:mix}

Suppose that $C,C'\in\C{n}{m}{\state}$. A convex combination $t C + (1-t)C'$, $0<t<1$, is a row-stochastic matrix, so an obvious question is whether it also belongs to $\C{n}{m}{\state}$ or not.

Firstly, we observe that if $C$ and $C'$ have implementations with $s,M$ and $s',M'$, respectively, and if $s=s'$ or $M=M'$, then $t C + (1-t)C'$ can be implemented with the corresponding mixture of state ensembles or measurements.

Secondly, if Alice and Bob are allowed to have shared randomness, then by having implementations $s,M$ and $s',M'$ they can implement $t C + (1-t)C'$ by using the implementations in a coordinated way. 
However, one cannot conclude that $t C + (1-t)C'\in\C{n}{m}{\state}$ as shared randomness is an additional resource that is not part of the definition of $\C{n}{m}{\state}$.

With Example \ref{ex:nonconvex} we demonstrate that $C,C'\in\C{n}{m}{\qubit}$ does \emph{not} imply that $t C + (1-t)C'\in\C{n}{m}{\qubit}$, meaning that $\C{n}{m}{\qubit}$ is not a convex set.
This also means that shared randomness is indeed an additional resource that can be used to implement some communication matrices that cannot be implemented without it.
A 
We start by proving a simple auxiliary result that is utilized in Example \ref{ex:nonconvex}  and can be used to generate other similar examples.

\begin{proposition}\label{prop:rank1}
If $C\in\C{n}{n}{\qudit}$, then $\tr{C}\leq d$.
Assuming that $C$ contains no zero columns, the equality holds only if $C$ has an implementation with pure states and rank-1 POVM (i.e. every operator has rank one). 
\end{proposition}

\begin{proof}
Suppose $C\in\C{n}{n}{\qudit}$ and hence $C_{ab}=\tr{\varrho_a \Mo(b)}$ for some states $\varrho_1,\ldots,\varrho_n$ and $n$-outcome POVM $\Mo$ in $\qudit$.
We have $C_{ab} \leq r_b$, where $r_b$ is the maximal eigenvalue of $\Mo(b)$.
It follows that
\begin{align*}
\tr{C} \leq \sum_b r_b \leq \sum_b \tr{\Mo(b)} = \tr{\id} = d \, .
\end{align*}
The second inequality is equality if and only if $\tr{\Mo(b)}=r_b$ for every $b$, and this means that $\Mo(b)$ is either rank-1 or zero.
Assuming that $C$ contains no zero columns, every $\Mo(b)$ must be rank-1.
It follows that $\varrho_b$ must be the unique eigenstate of $\Mo(b)$ with eigenvalue $r_b$.
\end{proof}

\begin{example}\label{ex:nonconvex}
Let us first notice that $D_{3,1/3}\in\C{3}{3}{\qubit}$.
Namely, we obtain this communication matrix e.g. by choosing
\begin{align*}
\varrho_1 = \begin{bmatrix}
1 & 0 \\
0 & 0
\end{bmatrix}, \, \varrho_2 =  \frac{1}{4} \begin{bmatrix}
1& \sqrt{3}\\
\sqrt{3} & 3
\end{bmatrix}, \, \varrho_3 =  \frac{1}{4} \begin{bmatrix}
1 & -\sqrt{3}\\
-\sqrt{3} & 3
\end{bmatrix}
\end{align*} and $\Mo(b)=\tfrac{2}{3} \varrho_b$, $b=1,2,3$.

We then observe that if a communication matrix $D\in\C{3}{3}{\qubit}$ has $\frac{2}{3}$ in all diagonal entries, then $D=D_{3,1/3}$.
To see this, we use Prop. \ref{prop:rank1} and hence write $D_{ab}= \tr{\varrho_a \Mo(b)}$ with $\Mo(b) = \tfrac{1}{3} (\id + \vv_b \cdot \vsigma)$, where $\vv_1,\vv_2,\vv_3$ are unit vectors in $\R^3$ and $\vsigma=(\sigma_x,\sigma_y,\sigma_z)$.
As $\sum_b \Mo(b)=\id$, the unit vectors are summing to $0$. 
It follows that they are determined up to a unitary transformation.
Further, we must have $\varrho_a = \tfrac{3}{2} \Mo(a)$ in order to have the maximal eigenvalues in the diagonal.
These facts imply that $D=D_{3,1/3}$.

Let us then construct two qubit communication matrices.
We choose qubit states 
\begin{align*}
\varrho_1 = \begin{bmatrix}
1 & 0 \\
0 & 0
\end{bmatrix}, \, \varrho_2 = \begin{bmatrix}
0 & 0\\
0 & 1
\end{bmatrix}, \, \varrho_3 = \begin{bmatrix}[1.2]
\frac 12 & -\frac{3}{10} + \frac{2i}{5} \\
-\frac{3}{10}- \frac{2i}{5} & \frac 12
\end{bmatrix}
\end{align*} and qubit measurement \begin{align*}
\Mo(1) = \frac 14 \begin{bmatrix}
1 & 1 \\
1 & 1
\end{bmatrix}, \, \Mo(2) = \frac 13 \begin{bmatrix}
1 & -i \\
i & 1
\end{bmatrix}, \, \Mo(3) = \begin{bmatrix}[1.2]
\frac{5}{12} & -\frac 14 + \frac i3 \\
- \frac 14 - \frac i3 & \frac{5}{12}
\end{bmatrix}.
\end{align*} 
With these choices we obtain the communication matrix 
\begin{align*}
C =  \tr{\varrho_i \Mo(j)}  = \begin{bmatrix}[1.2]
\frac 12 & \frac 13 & \frac 16 \\
\frac 14 & \frac 23 & \frac{1}{12} \\
\frac{1}{10} & \frac{1}{15} & \frac 56
\end{bmatrix}.\end{align*} 
We define a new set of states by setting $\varrho_1'=\varrho_3$, $\varrho_2'=\varrho_2$ and $\varrho_3'=\varrho_1$, and a new measurement $\Mo'$ as $\Mo'(1)=\Mo(3)$, $\Mo'(2)=\Mo(2)$ and $\Mo'(3)=\Mo(1)$.
With these choices we get the communication matrix
\begin{align*}
C' =  \begin{bmatrix}[1.2]
\frac 56 & \frac{1}{15} & \frac{1}{10} \\
\frac{1}{12} & \frac 23 & \frac 14 \\
\frac 16 & \frac 13 & \frac 12
\end{bmatrix} \, .
\end{align*} 
The equal mixture of $C$ and $C'$ gives
\begin{align*}
\half C + \half C'  = \begin{bmatrix}[1.2]
\frac 23 & \frac 15 & \frac{2}{15} \\
\frac 16 & \frac 23 & \frac 16 \\
\frac{2}{15} & \frac 15 & \frac 23
\end{bmatrix} \, .
\end{align*} 
This matrix has $\frac{2}{3}$ in all diagonal entries  but it is different than $D_{3,1/3}$.
By our earlier observation we conclude that $\half C + \half C' \notin \C{3}{3}{\qubit}$, hence $\C{3}{3}{\qubit}$ is not convex.
\end{example}

\section{Ultraweak matrix majorization}\label{sec:uw}

The concept of ultraweak matrix majorization was introduced in the present physical context in \cite{HeKe19}. 
In \cite{CSM98} this relation has been coined as I/O-degradation, where I-O abbreviates input and output.
In this section we recall this concept and develop it further.
The ultraweak matrix majorization gives a precise meaning when one communication task is more difficult than another one, or when they should be considered equally difficult. 

\begin{definition}
For two matrices $C\in\M{a}{b}$ and $D\in\M{c}{d}$, we denote $C\uleq D$ if there are row-stochastic matrices $L\in\Mr{a}{c}$ and $R\in\Mr{d}{b}$ such that $C=LDR$.
We then say that $C$ is \emph{ultraweakly majorized} by $D$.
We further denote $C \simeq D$ if $C\uleq D \uleq C$, and in this case we say that $C$ and $D$ are \emph{ultraweakly equivalent}.
\end{definition}

The relation $\uleq$ is reflexive and transitive, hence a preorder. 
It fails to be antisymmetric; we can have $C\simeq D$ although $C\neq D$.

The ultraweak majorization is weaker than the matrix majorization \cite{Dahl99} and weak matrix majorization \cite{PeMaSi05}; these are defined in the same way but in the matrix majorization $L=\id$ while in the weak matrix majorization $R=\id$.
We note that the matrix majorization makes sense for two matrices only when they have the same number of rows, while in the weak matrix majorization the  number of columns is assumed to be the same. 
However, in the ultraweak majorization matrices can have different number of rows and columns.

The reason to introduce and study ultraweak majorization in the current investigation is that $\Call{\state}$, the set of all communication matrices implementable with a given system $\state$, is closed with respect to ultraweak majorization.
Namely, as shown in \cite{HeKe19}, if $C \uleq D$ and $D\in\Call{\state}$, then $C\in\Call{\state}$.
In a physical setting the relation $C=LDR$ can be seen as follows. 
Suppose we have a communication setup that implements $D$.
If $C =LDR$, then by mixing and relabeling the states and measurement outcomes we can implement $C$, the respective operations being given by $L$ on states and $R$ on measurement outcomes.
For this reason, we take $C \uleq D$ to be the mathematical formulation of the idea that $C$ is easier or equally easy to implement than $D$.

\begin{figure}[t]
\begin{center}
\includegraphics[scale=0.5]{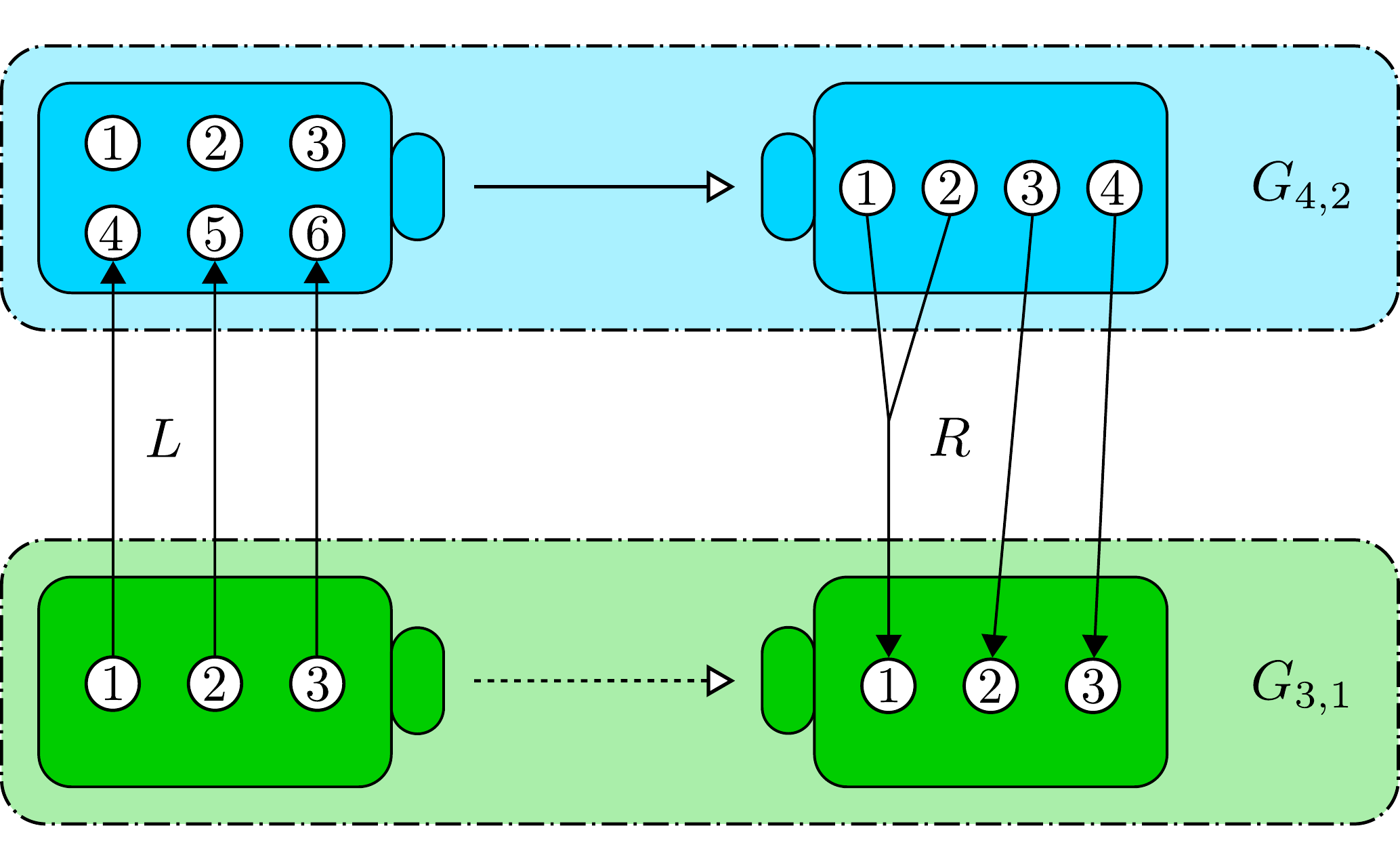}
\caption{\label{fig-1} The blue devices that implement the communication matrix $G_{4,2}$ can also be used to implement $G_{3,1}$ without the green devices. The majorization matrizes $L$ and $R$ written in \eqref{eq:G31} give us instructions how to achieve this by preprocessing the states and postprocessing the measurement outcomes. In this case we see that we need to only consider the last three preparations in the blue preparator and when measuring the system with the blue measurement device with four outcomes we need to combine the first two outcomes into a single new outcome in order to obtain the same statistics as one would get with the green one.}
\end{center}
\end{figure}

As an example, it has been shown in \cite{HeKe19} that $G_{n,t-1} \uleq G_{n,t} \uleq G_{n+1,t+1}$. 
The concrete case of $G_{3,1} \uleq G_{4,2}$ can be explicitly written as
\begin{equation}\label{eq:G31}
G_{3,1} = \half 
\begin{bmatrix}[1]
 1 & 1 & 0 \\
 1 & 0 & 1 \\ 
 0 & 1 & 1
\end{bmatrix} = 
\begin{bmatrix}
0 & 0 & 0 & 1 & 0 & 0 \\
0 & 0 & 0 & 0 & 1 & 0 \\
0 & 0 & 0 & 0 & 0 & 1 
\end{bmatrix}
\half \begin{bmatrix}
1 & 1 & 0 & 0 \\ 
1 & 0 & 1 & 0 \\ 
1 & 0 & 0 & 1 \\ 
0 & 1 & 1 & 0 \\ 
0 & 1 & 0 & 1 \\ 
0 & 0 & 1 & 1  
\end{bmatrix}
\begin{bmatrix}
1 & 0 & 0 \\
1 & 0 & 0 \\
0 & 1 & 0 \\
0 & 0 & 1 
\end{bmatrix}
= L G_{4,2} R,
\end{equation}
and this is illustrated in Fig. \ref{fig-1}.
It is easy to see that for any $n$ and $\epsilon\in[0,1]$, we have $V_n \uleq D_{n,\epsilon} \uleq \id_n$.
We will fully characterize the ultraweak ordering of $D_{n,\epsilon}$ matrices in Section \ref{sec:example}.   

Some simple sufficient conditions for ultraweak equivalence are listed below. The conditions (a)--(c) were presented already in \cite{CSM98}.
The condition (d) generalizes (b) while (e) generalizes (c).
The proofs are straightforward and we skip them.

\begin{proposition}\label{prop:equivalent}
In the following cases two matrices $C$ and $D$ are ultraweakly equivalent.
\begin{itemize}
\item[(a)] $C$ is obtained from $D$ by permuting rows and columns.
\item[(b)] $C$ is obtained from $D$ by duplicating some of the rows of $D$.
\item[(c)] $C$ is obtained from $D$ by adding zero columns to $D$.
\item[(d)] $C$ is obtained from $D$ by adding a row that is a convex mixture of some of the rows $D$.
\item[(e)] $C$ is obtained from $D$ by splitting a column into several columns with weights summing into $1$.
\end{itemize}
\end{proposition}

All of the conditions (a)--(e) have clear operational interpretations.
The physical operations that correspond to these conditions are the following:
\begin{itemize}
\item[(a)] Relabel bijectively states and measurement outcomes.
\item[(b)] Add a state that is identical to one of the existing states.
\item[(c)] Add the zero effect to the measurement device.
\item[(d)] Add a state that is a mixture of the existing states.
\item[(e)] Split an outcome of the measurement device into several outcomes, possibly with different probabilistic weights.
\end{itemize}
We note that as the ultraweak equivalence is a symmetric relation, these actions can be reverted. 
For instance, we can remove a zero column from a communication matrix and it is ultraweakly equivalent to the original one.

\section{Ultraweak monotones}\label{sec:monotones}

\begin{definition}
A function $f:\Mrall \to \R$ is an \emph{ultraweak monotone} if $C\uleq D$ implies $f(C) \leq f(D)$ for all $C,D\in\Mrall$.
\end{definition}

As with any monotones, ultraweak monotones can be used to verify that a matrix does not majorize another one; if $f(C) > f(D)$ for some $C,D\in\Mrall$, we can conclude that $C\npreceq D$.
Especially, if $f(C) \neq f(D)$, then $C\not\simeq D$.
In the latter case we say that $f$ \emph{detects the inequivalence} of $C$ and $D$.

For a given monotone $f$, we are also interested on the maximal value of $f$ in a given operational theory. 
As we will see, these maximal values may give operationally motivated `dimensions' of $\state$ and can be used to analyze the dissimilarities between different theories.

\subsection{Rank}

The rank is an ultraweak monotone as the rank cannot increase in matrix multiplication.
The rank seems not to have a direct operational meaning.
However, the maximal rank in $\state$ links to the linear dimension in the following way and gives a convenient necessary condition to test if $C\in\Call{\mathcal{S}}$.

\begin{proposition}\label{prop:maximal-rank}
Let $\mathcal{S}$ be a $d$-dimensional state space. 
Then the maximal rank of $C\in\Call{\mathcal{S}}$ is $d+1$.
\end{proposition}

\begin{proof}
We first show that there exists a communication matrix with rank of $d+1$ and then proceed to show that every other communication matrix has a smaller or equal rank.

If $\state \subset \mathcal{V}$ is $d$-dimensional, i.e., $\dim(\mathrm{aff}(\state))=d$, where $\mathrm{aff}(\state)$ denotes the affine hull of $\state$, then we can take the vector space $\mathcal{V}$ to be $(d+1)$--dimensional and embed $\state$ in $\mathcal{V}$ such that $\mathrm{aff}(\state)$ is an affine hyperplane in $\mathcal{V}$ and $\state$ spans $\mathcal{V}$. It then follows that $\effect(\state) \subset \mathcal{V}^*$, where $\mathcal{V}^*$ denotes the dual of $\mathcal{V}$ and $\dim(\mathcal{V}^*) = \dim(\mathcal{V}) = d+1$, and furthermore one can confirm that $\effect(\state)$ spans $\mathcal{V}^*$.

Since $\state$ is $d$-dimensional there exists a maximal set of affinely independent states $\{s_1, \ldots, s_{d+1} \}$ such that 
\begin{equation}\label{eq:aff}
\mathrm{aff}\left(\{s_1, \ldots, s_{d+1} \} \right) = \mathrm{aff}(\state) \subset \mathcal{V}.
\end{equation}
Since $\mathrm{aff}(\state)$ is an affine hyperplane, there exists a functional $u \in \mathcal{V}^*$ such that $\mathrm{aff}(\state)= \{x \in \mathcal{V}\, | \, u(x) = 1 \}$. Thus, in particular $u$ is an interior point of the set of positive functionals on $\state$ and since $\effect(\state)$ is a subset of those and $\effect(\state)$ spans $\mathcal{V}^*$, there exists $d+1$ linearly independent effects $\{M_j\}_{j=1}^{d+1}$ such that $\sum_j M_j = u$ so that they form a measurement $M$ on $\state$. We consider the $(d+1) \times (d+1)$ communication matrix $C$ with $C_{ij} = M_j(s_i)$ for all $i,j= 1, \ldots, d+1$. Let us consider the linear dependence of the columns of $C$, i.e., let $\{\alpha_j\}_{j=1}^{d+1} \subset \R$ such that 
\begin{equation*}
\sum_{j=1}^{d+1} \alpha_j \left(M_j(s_1), \ldots, M_j(s_{d+1})\right) = \vec{0},
\end{equation*}
which implies that $\sum_j \alpha_j M_j(s_i)=0$ for all $i=1, \ldots, d+1$. Let $s \in \state$ be any state. By Eq. \eqref{eq:aff}, there exists real numbers $\{\mu_i\}_{i=1}^{d+1} \subset \R$ with $\sum_i \mu_i =1$ such that $s= \sum_i \mu_i s_i$. By applying the functional $\sum_j \alpha_j M_j$ on $s$, we find that $\sum_j \alpha_j M_j(s) = 0$. Because $s$ was an arbitrary state, and because $\state$ spans $\mathcal{V}$, we have that $\sum_j \alpha_j M_j(x) = 0$ for all $x \in \mathcal{V}$. Thus, $\sum_j \alpha_j M_j = 0$ and since $\{M_j\}_j$ is linearly independent, we must have that $\alpha_1= \cdots = \alpha_{d+1}=0$. Thus, $\rank{C} = d+1$.

Let then $D \in \Call{\state}$ be any $n\times m$ communication matrix generated by a $m$-outcome measurement $M'$ and states $\{s'_1, \ldots,s'_n\}$. If $m \leq d+1$ or $n \leq d+1$, then obviously $\rank{D} \leq d+1$. Let then $m> d+1$. Because $\dim(\mathcal{V}^*)=d+1$, the effects of $M'$ cannot be linearly independet, so that there exists an index $k \in \{1, \ldots,m\}$ such that $M'_k = \sum_{j\neq k} \beta_j M'_j$ for some set of real numbers $\{\beta_j\}_j$ of which not all of them are zero. Thus, looking at the $k$th column of $D$, we find that 
\begin{equation*}
\left(M'_k(s'_1), \ldots, M'_k(s'_n)\right) = \sum_{j\neq k} \beta_j \left(M'_j(s'_1), \ldots, M'_j(s'_n)\right),
\end{equation*}
so that the columns are not linearly independent. The same procedure can be continued to show that $D$ can have maximally $d+1$ linearly independent columns, i.e., $\rank{D} \leq d+1$.
\end{proof}

We can thus define the linear dimension $d_{lin}(\state)$ of a theory with a state space $\state$ as
\begin{align*}
d_{lin}(\state) := \sup\{ \rank{C} \, | \, C \in \Call{\state} \}.
\end{align*}
From the proof of the previous proposition we find the supremum is always attained and that if $\state$ is $d$-dimensional, i.e., $\dim(\mathrm{aff}(\state)) = d$, then $d_{lin}(\state) = d+1$. 

For example, for $d$-dimensional classical and quantum theory we have that $d_{lin}(\state^{cl}_d) = d$ and $d_{lin}(\qudit)= d^2$. We note that in this context, when we call classical or quantum state spaces $d$-dimensional, we mean that their operational dimension $d_{op}(\state)$, defined as the maximal number of distinguishable states, or equivalently, the maximal number $n$ such that $\id_n \in \Call{\state}$, is $d$. We note that for quantum state spaces the operational dimension coincides with the dimension of the underlying Hilbert space. We will see later how operational dimension itself is linked to a particular monotone.

\begin{example}
Considering the uniform (anti)distinguishability matrices introduced in Sec. \ref{sec:some}, we have $\rank{D_{n,\epsilon}}=n$ for all $\epsilon\neq 1-1/n$, and $\rank{D_{n,1-1/n}}=1$.
Therefore, for a fixed $n$ the rank shows that $D_{n,1-1/n} \ul D_{n,\epsilon}$ for $\epsilon\neq 1-1/n$ but nothing else than that.
For $n<m$ the rank shows that $\id_n$ does not ultraweakly majorize any $D_{m,\epsilon}$ with $\epsilon\neq 1-1/m$.
\end{example}

\subsection{Nonnegative rank}

We recall that a matrix $C$ is called \emph{nonnegative} if $C_{ij}\geq 0$ for all $i,j$.
The \emph{nonnegative rank} of a nonnegative matrix $C \in \M{n}{m}$, denoted as $\nrank{C}$, is defined as the the smallest number $k$ such there exists nonnegative matrices $L \in\M{n}{k} , R\in\M{k}{m}$ such that $LR=C$. 
As shown in \cite{CoRo93}, for a stochastic matrix $C\in\Mr{n}{m}$, the nonnegative matrices $L$ and $R$ can be chosen to be stochastic so that for $C \in \Mr{n}{m}$ we have
\begin{align} \label{eq:nrank1}
\nrank{C} = \min \{ k\in\nat \, | \,\exists L \in\Mr{n}{k} , R\in\Mr{k}{m}:  C=LR \} \, .
\end{align}
By noting that $C = LR = L \id_k R$, we see that we can express the nonnegative rank with respect to ultraweak matrix majorization as follows:
\begin{align}\label{eq:nrank2}
\nrank{C} = \min \{ k\in\nat \, | \, C \uleq \id_k  \}. 
\end{align}
From this expression it is obvious that $\nrank{\cdot}$ is an ultraweak monotone.

In comparison to the ordinary rank, we see from \eqref{eq:nrank1} that 
\begin{align}\label{eq:rank-ineq}
\rank{C} \leq \nrank{C} \leq \min(n,m) \, .
\end{align}
It follows that any $n\times n$ square matrix $C$ with $\rank{C}=n$ has $\nrank{C}=n$.

\begin{example}[Nonnegative rank of $G_{n,t}$] \label{ex:G42-nrank}
It has been shown in \cite{HeKe19} that $\rank{G_{n,t}}=n$ for all $1\leq t \leq n-1$. 
Since $G_{n,t}$ is a ${{n}\choose{t}} \times n$ matrix,  we see from Eq. \eqref{eq:rank-ineq} that $\nrank{G_{n,t}} = n$.
\end{example}

It was shown in \cite{CoRo93} that for a nonnegative matrix $A \in \M{n}{m}$ such that $\rank{A} \leq 2$, we have that $\nrank{A} = \rank{A}$. 
Especially, this with \eqref{eq:rank-ineq} means that $\nrank{A}=2$ if and only if $\rank{A}=2$.
Furthermore, it was proven in \cite{CoRo93} that if $A \in \M{n}{m}$ such that $n \in \{1,2,3\}$ or $m \in \{1,2,3\}$, then $\nrank{A} = \rank{A}$.
For a $4\times 4$ matrix we can already have $\nrank{C}>\rank{C}$; if \begin{align*}
K=\frac{1}{2} \begin{bmatrix}
1 & 1 & 0 & 0 \\
1 & 0 & 1 & 0 \\
0 & 1 & 0 & 1 \\
0 & 0 & 1 & 1 \\
\end{bmatrix}
\end{align*}
then $\rank{K}=3$ but $\nrank{K}=4$ \cite{CoRo93}.

The nonnegative rank of a communication matrix $C$ has an operational interpretation as the smallest classical system that can implement $C$; this is the content of the following proposition.
We recall that $\cld$ is the state space of a classical system with $d$ pure states, i.e., $d-1$ -simplex.

\begin{proposition}\label{prop:classical_communication_matrices}
For $C\in\Mrall$, the following are equivalent:
\begin{itemize}
\item[(i)] $\nrank{C}\leq n$
\item[(ii)] $C \uleq \id_n$
\item[(iii)] $C\in \Call{\mathcal{S}^{cl}_n}$
\end{itemize}
\end{proposition}

\begin{proof}
We have already seen that (i)$\Leftrightarrow$(ii).
The implication (ii)$\Rightarrow$(iii) follows from the fact that $\id_n\in\Call{\state^{cl}_n}$. 
It remains to show (iii)$\Rightarrow$(ii).
A classical theory $\mathcal{S}^{cl}_n$ has exactly $n$ pure states $s_1,\ldots,s_n$ and any mixed state has a unique convex decomposition into the pure states.
Further, any measurement is a post-processing of the $n$-outcome measurement $M$, defined as $M_b(s_a) = \delta_{ab}$ for all $a,b=1,\ldots,n$.
It follows that any communication matrix that has an implementation in $\cln$ can be obtained from $\id_n$ with classical processing on states and measurement outcomes. 
\end{proof}

Any communication matrix $C$ is ultraweakly majorized by $\id_n$ for some $n$; if $C\in\Mr{n}{m}$ we have $C\uleq\id_n$ and $C\uleq\id_m$. 
It makes thus sense to look for the minimal classical theory that can be used to implement all communication matrices in $\Call{\state}$. 
Thus, we set
\begin{align*}
d_{cl}(\state) &:=  \inf \{ k \in \nat \, | \, \Call{\state} \subseteq \Call{\state^{cl}_k} \}
\end{align*}
and by the previous proposition it follows that
\begin{align*}
d_{cl}(\state) &=  \inf \{ k \in \nat \, | \, \forall C \in \Call{\state}: C \uleq \id_k \} \\
&= \inf \{ k \in \nat \, | \, \forall C \in \Call{\state}: \nrank{C} \leq k \} \\
&= \sup\{ \nrank{C} \, | \, C \in \Call{\state} \}.
\end{align*}
We note that if $d_{cl}(\state) < \infty$, then the supremum in the last expression is always attained.

Compared with the operational dimension, we immediately see that $d_{cl}(\state) \geq d_{op}(\state)$ as if $\id_d \in \Call{\state}$ for any $d$, then $\nrank{\id_d} = d$. 
On the other hand, we observe that in general $d_{cl}(\state) \neq d_{op}(\state)$: let $\state= \qubit$, so that $d_{op}(\qubit) = 2$, but since $A_4 \in \Call{\qubit}$ and $\nrank{A_4}=4$, we have that $d_{cl}(\qubit)\geq 4$. 


\begin{remark}
In \cite{DaBrToBuVe17} the analogue of $d_{cl}(\state)$ was called the \emph{signalling dimension} of $\state$, but in their prepare-and-measure scenario also shared randomness was included in the protocol. Thus, with shared randomness the signalling dimension $d_s(\state)$ of a state space $\state$ can be expressed as 
\begin{align*}
d_{s}(\state) &:=  \min \{ k \in \nat \, | \, \forall m,n \in \nat: \Csr{n}{m}{\state} \subseteq \Csr{n}{m}{\state^{cl}_k} \} \, .
\end{align*}
In \cite{FrWe15} it was shown that $\Csr{n}{m}{\qudit} = \Csr{n}{m}{\state^{cl}_d} $ for all $n,m,d \in \nat$ so that $d_s(\qudit) = d_s(\state^{cl}_d) = d = d_{op}(\qudit) = d_{op}(\state^{cl}_d)$. 
As we saw earlier, without shared randomness (i.e. without the convex hulls) we have $\C{n}{m}{\qudit} \neq \C{n}{m}{\state^{cl}_d}$ and therefore $d_{cl}(\qudit) \neq d_{cl}(\state^{cl}_d)$.
\end{remark}

\subsection{Positive semidefinite rank}

The \emph{positive semidefinite rank} of a nonnegative matrix $C \in \M{n}{m}$, denoted as $\psdrank{C}$, is defined as the smallest integer $k$ such that there exist positive semidefinite $k\times k$ matrices $A_1,\ldots,A_n$ and $B_1,\ldots,B_m$ such that $C_{ij}=\tr{A_i B_j}$.
This is called a \emph{positive semidefinite decomposition} of $C$.

\begin{proposition}
The positive semidefinite rank $\psdrank{\cdot}$ is an ultraweak monotone.
\end{proposition}

\begin{proof}
Let $C,D\in\Mrall$ and $C\uleq D$. 
Let us assume that $D$ has a positive semidefinite decomposition $D_{ij}=\tr{A_i B_j}$.
Since $C\uleq D$, there exist $L,R\in\Mrall$ such that $C=LDR$.
Then
\begin{align*}
C_{k\ell} &= \sum_{i,j} L_{ki} D_{ij} R_{j\ell} = \sum_{i,j} L_{ki} \tr{A_iB_j} R_{j\ell} \\
&= \tr{(\sum_i L_{ki}A_i)(\sum_j R_{j\ell}B_j)} = \tr{A'_kB'_\ell} \, ,
\end{align*}
where $A'_k=\sum_i L_{ki}A_i$ and $B'_\ell=\sum_j R_{j\ell}B_j$. 
The matrices $A'_k$ and $B'_j$ are positive semidefinite matrices of the same size as $A_i$ and $B_j$.
\end{proof}

In \cite{LeWeWo17} it was shown that $\psdrank{C}$ gives the dimension of the quantum system needed to produce the communication matrix $C$; the following result should be compared with Prop. \ref{prop:classical_communication_matrices}.

\begin{proposition}\label{prop:psdrank-quantum}\cite[Lemma 5]{LeWeWo17} Let $C \in \Mrall$. Then $C \in \Call{\qudit}$ if and only if $\psdrank{C} \leq d$.
\end{proposition}

We denote
\begin{align*}
d_q(\state) &:= \inf\{ d  \in \nat \, | \, \Call{\state} \subseteq \Call{\qudit} \}
\end{align*}
and then $d_q(\state)$ is the minimum dimension of a quantum system needed to produce all of the communication matrices on $\state$.
We see from Prop. \ref{prop:psdrank-quantum} that
\begin{align*}
d_q(\state) &= \inf \{d \in \nat \, | \, \forall C \in \Call{\state}: \psdrank{C} \leq d \} \\
&= \sup \{ \psdrank{C} \, | \, C \in \Call{\state} \}.
\end{align*}
Similarly to the classical dimension, if the quantum dimension is finite, then supremum in the last expression is attained for some communication matrix.

It has been shown in \cite{GoPaTh13} that
\begin{align}\label{eq:psd-rank-ineq}
\sqrt{\rank{C}} \leq  \psdrank{C} \leq \nrank{C} \, .
\end{align}
Namely, the second inequality follows from the observation that the nonnegative rank corresponds to the minimal positive semidefinite decomposition with the extra requirement that the matrices in the decomposition are diagonal. 
The first inequality is easy to confirm in the case when $C$ is a stochastic matrix; if $C \in \Mrall$ and $\psdrank{C} = d$, then from Prop. \ref{prop:psdrank-quantum} it follows that $C \in \Call{\qudit}$ so that by Prop. \ref{prop:maximal-rank} we must have $\rank{C} \leq d^2$ from which the claim follows.

In the following we derive a lower bound for communication matrices of a special type.

\begin{proposition}\label{prop:psdrank-kernel}
Let $C$ be a communication matrix such that all rows are different.
If there is a column with at least two zeros and at least one nonzero element, then $\psdrank{C}\geq 3$.
\end{proposition}

\begin{proof}
Suppose that $\psdrank{C} = 2$, i.e.,  there exists a quantum implementation such that $C_{ij} = \tr{\varrho_i \Mo(j)}$, where the states and effects are 2-dimensional. As all rows of $C$ are different we know that $\varrho_i \neq \rho_j$ for $i \neq j$. Let $k$ be the column with (at least) two zeros, say $\tr{\varrho_i \Mo(k)} = \tr{\varrho_j \Mo(k)} = 0$. We also know that there exists an index $h$ such that $\tr{\varrho_h \Mo(k)} \neq 0$, which means that $\Mo(k) \neq 0$. However, a 2-dimensional nonzero effect has at most 1-dimensional kernel, so it is impossible for both $\tr{\varrho_i \Mo(k)}$ and $\tr{\varrho_j \Mo(k)}$ to be zero.
\end{proof}

\begin{example}\label{ex:G42-psdrank}
We note that $\rank{A_3}=\nrank{A_3}=3$ but $\psdrank{A_3}=2$. More generally, it is known that $\psdrank{A_{d^2}} = d$ whenever $d$ is odd and $\psdrank{A_{d^2-1}} = d$ whenever $d$ is even \cite{LeWeWo17}. 
It is also known that there can be at most $d^2$ uniformly antidistinguishable qudit states \cite{HeKe19}. Note that the psd-factorization of, say, $A_n$ also gives a psd-factorization for $A_{n-1}$. This is due to the fact that we can drop one state and one effect from the decomposition and obtain a psd-decomposition for $A_{n-1}$ of the same size as for $A_n$. It is therefore guaranteed that if $A_n$ has a qudit implementation, then $A_{n-1}$ also has a qudit implementation of the same dimension. We can hence make the following chain of conclusions. Let $d$ be odd. Then $\psdrank{A_{d^2}}$ is $d$. Now every $\psdrank{A_{d^2-k}} = d$ until reaching a $k$ such that $d^2-k = (d-1)^2$. Now $d-1$ is even so it is not known whether $\psdrank{A_{(d-1)^2}}$ is $d$ or $d-1$. However, as stated above, it is known that $\psdrank{A_{(d-1)^2-1}} = d-1$ and this chain can be continued until reaching $A_2$. 
We can also make the remark that, as was shown in \cite{HeKe19}, a SIC-POVM always gives an implementation for uniform antidistinguishability. As SIC-POVMs are known to exist for up to $d=151$ \cite{FuHoSt17}, we know that $\psdrank{A_{d^2}} = d$ for up to $d=151$ regardless of whether $d$ is even or odd.

For the matrices $G_{n,t}$ we know that $\psdrank{G_{n,t}} \leq n$ because $G_{n,t} \uleq \id_n$. For $1<t<n$ we are not aware of other useful bounds,  besides Eq. \eqref{eq:psd-rank-ineq}. From Eq. \eqref{eq:psd-rank-ineq} together with Example \ref{ex:G42-nrank} and Prop. \ref{prop:psdrank-kernel} we see that $\psdrank{G_{4,2}} \in \{3,4\}$. In Ref. \cite{HeKe19} it was shown that $\psdrank{G_{4,2}}=3$.
\end{example}

We can develop the method of obtaining a bound for the positive semidefinite rank from a larger matrix a bit further. Let $C \in \mathcal{M}^{row}_{a,b}$ be a communication matrix and $\psdrank{C}=n$. We see that if $D$ is a communication matrix such that $tD$ is a submatrix of $C$ for some $t>0$, then $\psdrank{D} \leq n$. Namely, it is clear that the positive semidefinite factorization for $D$ is obtained by dropping some matrices from the factorization $C_{ij} = \tr{A_i B_j}$ and renormalizing by a suitable constant $t$. This then quarantees a positive semidefinite factorization, and hence a quantum implementation, of the same size for $D$ as for $C$. As was stated in the previous example, this is exactly the case with $A_n$ and $A_{n-1}$, since $(n-2)/(n-1) A_{n-1}$ is a submatrix of $A_n$, where any row $j$ and column $j$ of $A_n$ can be dropped in order to obtain $(n-2)/(n-1) A_{n-1}$. 

\begin{example}
For a communication matrix $C$, let us define  $\lambda_{max}(C) := \sum_j \max_i (C_{ij})$. In \cite{LeWeWo17} it was shown that 
\begin{equation}\label{eq:psdrank-lmax}
\psdrank{C} \geq \lmax{C}.
\end{equation}
 We are interested in the case when a communication matrix $C$ does not have a qubit implementation, i.e., when $\psdrank{C} \geq 3$. By the previous inequality, if $\lmax{C}> 2$, then $\psdrank{C}\geq 3$. Let us next demonstrate how Prop. \ref{prop:rank1}  can be used to get another lower bound for a specific class of communication matrices $C \in \Call{\state}$ such that $\lmax{C}=2$. We demonstrate this with an example.

Let us consider two communication matrices
\begin{equation*}
C = \frac{1}{4} \begin{bmatrix}
2 & 1 & 1 & 0 \\
0 & 2 & 1 & 1 \\
1 & 0 & 2 & 1 \\
1 & 1 & 0 & 2 
\end{bmatrix}, \quad
D = \frac{1}{4} \begin{bmatrix}
2 & 0 & 1 & 1 \\
0 & 2 & 1 & 1 \\
1 & 1 & 2 & 0 \\
1 & 1 & 0 & 2 
\end{bmatrix}.
\end{equation*}
Suppose that $C \in \Call{\qubit}$. Since $\lmax{C} = \tr{C} =2$, Prop. \ref{prop:rank1} shows that in this case $C$ must have an implementation with a $4$-outcome rank-1 POVM $\Mo$ and $4$ pure states $\{\varrho_1, \varrho_2, \varrho_3, \varrho_4\}$ such that for each effect $\Mo(j)$ we have that $\varrho_j$ is the unique eigenstate of $\Mo(j)$ with eigenvalue $r_j=1/2$.  Let us consider the first column of $C$. Since $\Mo(1)$ is rank-1 and the eigenstates corresponding to the maximal and minimal eigenvalues of $\Mo(1)$ are  $\varrho_1$ and $\varrho_2$ respectfully, we must have that $\varrho_1$ and $\varrho_2$ are antipodal states on the Bloch spehere. Since the eigenstate corresponding to the maximal eigenvalue of $\Mo(2)$ is $\varrho_2$, and since $\varrho_1$ and $\varrho_2$ are antipodal, we should have that also $\Mo(1)$ and $\Mo(2)$ have antipodal Bloch representation so that $\varrho_1$ is the eigenstate of $\Mo(2)$ corresponding to the eigenvalue zero and $C_{12}=0$. However, this is not the case and thus $C \notin \Call{\qubit}$ and $\psdrank{C} \geq 3$. We note that the same arguments can be applied to square communication matrices similar to $C$.

However, for $D$ we see that the same arguments do not apply and indeed one has a qubit implementation for $D$ using two pairs of antipodal effects/states:
\begin{align*}
\Mo(1) = \frac{1}{4}(\id_2 + \sigma_x ), \quad \Mo(2) = \frac{1}{4}(\id_2 - \sigma_x)  \\
\Mo(3) = \frac{1}{4}(\id_2 + \sigma_y ), \quad \Mo(4) = \frac{1}{4}(\id_2 - \sigma_y) 
\end{align*}
with $\varrho_i = 2 \Mo(i)$ for all $i \in \{1,2,3,4\}$. Thus, $\psdrank{D}=2$.
\end{example}

\subsection{Min and Max monotones}\label{subsec:minmax}

In the following we introduce two easily computable ultraweak monotones.

\begin{proposition}
The functions
\begin{align*}
\lambda_{max}(C) := \sum_j \max_i (C_{ij}), \qquad \lambda_{min}(C) := -\sum_j \min_i (C_{ij})
\end{align*}
are ultraweak monotones.
\end{proposition}

\begin{proof}
Let $C \uleq D$ so that there exist $L,R\in\Mrall$ such that $C = L D R$.
For $\lambda_{max}$ we then have that
\begin{align*}
\lambda_{max}(C) &= \sum_j \max_i C_{ij} = \sum_j \max_i \left( \sum_{m,n} R_{mj} L_{in} D_{nm} \right) \\
&= \sum_{j,m} R_{mj} \left[ \max_i \left(  \sum_n L_{in} D_{nm} \right) \right] = \sum_m \max_i \left( \sum_n L_{in} D_{nm} \right) \\
&\leq \sum_m \max_i \left( \sum_n L_{in} \max_k D_{km} \right) = \sum_m \max_i \max_k D_{km} \\
&= \sum_m \max_k D_{km} = \lambda_{max}(D) \, .
\end{align*}
The proof for $\lambda_{min}$ is analogous.
\end{proof}

Similarly to the previous monotones, we can try to use $\lambda_{max}$ and $\lambda_{min}$ to give us insight to not just some particular communication matrix but to a whole theory. 
This also clarifies their physical meaning.
We first set
\begin{align}
\lambda_{max}(\state) &:= \sup \{ \lambda_{max}(C) \, | \, C \in \Call{\state} \} \, .
\end{align}
We denote by $\obss$ the set of all measurements in $\state$ and by $\Omega_M \subset \nat$ the (finite) outcome set of a measurement $M \in \obss$. Now we can reformulate the previous quantity as follows:
\begin{align}
\lambda_{max}(\state) &= \sup \{ \lambda_{max}(C) \, | \, \exists M \in \obss, \{ s_i \}_i \subset \state : C_{ij} = M_j(s_i) \ \forall i,j \}  \nonumber \\
&= \sup \left\lbrace \sum_{j \in \Omega_M} \max_i M_j(s_i) \, | \, M \in \obss, \{ s_i \}_i \subset \state \right\rbrace \nonumber \\
&= \sup \left\lbrace \sum_{j \in \Omega_M} \max_{s \in \state} M_j(s) \, | \, M \in \obss \right\rbrace .
\end{align}

Let us denote $\lmax{M} := \sum_{j \in \Omega_M} \max_{s \in \state} M_j(s)$ for a measurement $M \in \obss$ so that $\lambda_{max}(\state) = \sup_{M \in \obss} \lmax{M}$. 
For a given measurement $M$, the quantity $\lmax{M}$ is related to some minimal error discrimination and decoding tasks: If Alice encodes $n$ equally likely messages into same amount of states and Bob decodes this by performing a $n$-outcome measurement, Bob's probability of error $P^n_E$ is bounded by 
\begin{equation}
P^n_E \geq 1- \dfrac{\lambda_{max,n}(\state)}{n},
\end{equation}
where 
\begin{equation*}
\lambda_{max,n}(\state) = \sup \{ \lmax{M} \, | \, M \in \obss: \ \#\Omega_M \leq n \}.
\end{equation*}
Thus, for a given measurement $M$ with $n$ outcomes, we can interpret $\lmax{M}/n$ as the decoding power of $M$ as it tells us the optimal decoding probability that can be achieved using $M$. Furthermore, $\lambda_{max, n}(\state)/n$ gives us the optimal success probability for a minimal error discrimination of $n$ states in the whole theory. 
We note that 
\begin{equation*}
\lmax{\state} = \sup_{M \in \obss} \lmax{M} = \sup_{n \in \nat} \lambda_{max, n}(\state) \, .
\end{equation*}

\begin{example}\label{ex:lmax-quantum}
In quantum theory all implementable communication matrices have the form
\begin{equation}
C_{ij}= \tr{\varrho_i \Mo(j)} \, ,
\end{equation} 
where $\varrho_i$ are density matrices and $\Mo$ is a POVM.
It follows that for $\qudit$ we have the upper bound
\begin{equation}
\lambda_{max}(C)= \sum_j \max_i \tr{\varrho_i \Mo(j)} \leq \sum_j \tr{\Mo(j)}=d \, ,
\end{equation} 
so that $\lmax{\qudit} \leq d= d_{op}(\qudit)$. Thus, we get the \emph{basic decoding theorem} \cite{QPSI10} in $d$-dimensional quantum theory: $P^N_E \geq 1- \frac{d}{N}$.
\end{example}

\begin{remark}\label{lambdamax}
In \cite{MaKi18}, $\lmax{\state}$ was called the \emph{information storability of $\state$} (denoted by $\mathfrak{n}$ in their work), and they were able to prove that it is related to the point-asymmetry of the state space: the information storability $\mathfrak{n}$ of $\state$ is related to the amount of asymmetry of $\state$ given by the (affine-invariant) Minkowski measure $\mathfrak{m}$ by the relation $\mathfrak{m} = \mathfrak{n}-1$. Since the Minkowski measure $\mathfrak{m}$ gives value $1$ only for point-symmetric state spaces, their result showed that $\lmax{\state} = \mathfrak{n} =2$, i.e., the state space can store 1 bit of information if and only if the state space is point-symmetric. 
\end{remark}

As was shown in Example \ref{ex:lmax-quantum}, in quantum theory $\lmax{\qudit}$ is bounded by the operational dimension $d$. The same can be shown for $d$-dimensional classical theory as well. By the results of \cite{MaKi18}, we see that this is not the case in general. In fact they show that in the so-called pentagon state space $\state_5$ we have $\lmax{\state_5} = 1+ 1/\cos(\pi/5) \approx 2.24 >2 = d_{op}(\state_5) $.

We note that by considering $\lambda_{min}$ similar way for the whole theory is not as useful since $\lambda_{min}(C) \in [-1,0]$ for all $C \in \Call{\state}$, where the minimum value $-1$ is obtained by any trivial measurement and the maximum value $0$ is obtained by any $\id_k$ or $A_n$ such that $\id_k, A_n \in \Call{\state}$.

\subsection{Distinguishability monotone}

For $C\in\Mrall$, we define
\begin{equation*}
\iota(C) := \max \{ n \, | \, \id_n \uleq C \} \, . 
\end{equation*}
It is obvious that $\iota$ is an ultraweak monotone.
Physically, $\iota$ describes how many messages we can send if we can implement $C$.
In $\Call{\state}$, the minimal value of $\iota$ is $1$ and the maximal value is $d_{op}(\state)$. 

\begin{proposition} \label{prop:iota}
For $C \in \Mrall$ we have that $\iota(C)=k$ if and only if the maximal number of orthogonal rows of $C$ is $k$.
\end{proposition}
\begin{proof}
First we will show that $C \in \Mr{n}{m}$ has $k$ orthogonal rows if and only if there exist stochastic matrices $L \in \Mr{k}{n}$ and $R \in \Mr{m}{k}$ such that $LCR = \id_k$, and then we will deal with the maximality of $k$ afterwards. 
For a matrix $A$ we use the notation $\vec{A}_i$ for the vector consisting of the elements of the $i$th row of $A$. 

Let first $C \in \Mr{n}{m}$ have $k$ orthogonal rows, i.e., there exists indices $\{p_1, \ldots, p_k\} \subset \{1, \ldots, n\}$ such that $\vec{C}_{p_i} \cdot \vec{C}_{p_j} = 0$ for all $i \neq j$, $i,j \in \{1, \ldots, k\}$. For each $i \in \{1, \ldots, k\}$ we define $Q_i = \{ q \in \{1, \ldots, m\} \, | \, C_{p_i q} \neq 0 \}$. We see that if $q \in Q_i$, then for all $j \in \{1, \ldots,k\}$, $j\neq i$, we have that
\begin{align*}
0 = \vec{C}_{p_i} \cdot \vec{C}_{p_j} = \sum_{q'=1}^m C_{p_iq'} C_{p_jq'} = C_{p_iq} C_{p_jq} +  \sum_{q' \neq q} C_{p_iq'} C_{p_jq'}
\end{align*}
so that $C_{p_jq}=0$ and therefore $q \notin Q_j$. Thus, $Q_i \cap Q_j = \emptyset$ for all $i \neq j$. On the other hand we might have that $C_{p_i q} = 0$ for all $i$ so that $q \notin Q_i$ for any $i$, and thus it might happen that $\cup_{i=1}^k Q_i \neq \{1, \ldots, m\}$. 

Thus, we can define a matrix $R \in \M{m}{k}$ by 
\begin{align*}
R_{qj} = 
\begin{cases}
1, & \mathrm{if} \ q \in Q_j, \\
\frac{1}{k}, & \mathrm{if} \ q \notin \cup_{i=1}^k Q_i, \\
0 & \mathrm{otherwise}.
\end{cases}
\end{align*}
and see that based on the previous properties we have that $R \in \Mr{m}{k}$. Also we define another stochastic matrix $L \in \Mr{k}{n}$ by setting $L_{ip} = \delta_{p_i p}$ for all $i \in \{1, \ldots, k\}$ and $p \in \{1, \ldots, n\}$. Now we see that
\begin{align*}
(LCR)_{ij} & = \sum_{p=1}^n \sum_{q=1}^m L_{ip}C_{pq}R_{qj} = \sum_{p=1}^n \sum_{q=1}^m \delta_{p_i p} C_{pq} R_{qj}\\
& = \sum_{q=1}^m C_{p_i q} R_{qj} = \sum_{q \in Q_i} C_{p_i q} R_{qj} = \delta_{ij} 
\end{align*}
for all $i,j \in \{1, \ldots, k\}$. Hence $LCR = \id_k$.

Let then $LCR= \id_k$ for some stochastic matrices $L \in \Mr{k}{n}$ and $R \in \Mr{m}{k}$. Since $\id_k$ has orthogonal rows, we have for $i, j \in \{1, \ldots,k\}$, $i\neq j$, that 
\begin{align}
\nonumber 0 &= \overrightarrow{ (\id_k)}_i \cdot  \overrightarrow{ (\id_k)}_j = \overrightarrow{ LCR}_i \cdot  \overrightarrow{ LCR}_j = \sum_{x=1}^k (LCR)_{ix} (LCR)_{jx}  \\
\nonumber &= \sum_{x=1}^k \sum_{p,p'=1}^n \sum_{q,q'=1}^m L_{ip}C_{pq}R_{qx} L_{jp'}C_{p'q'}R_{q'x}  \\
\nonumber &= \sum_{q,q'=1}^m \left[ \left(\sum_{p=1}^n L_{ip}C_{pq} \right) \left(\sum_{p'=1}^n L_{jp'}C_{p'q'} \right) \right] \left( \sum_{x=1}^k R_{qx} R_{q'x} \right) \\
&= \sum_{q,q'=1}^m \left[(LC)_{iq}  (LC)_{jq'} \right] \left( \vec{R}_q \cdot \vec{R}_{q'} \right). \label{eq:id_i-cdot-id_j}
\end{align}
Now if in the sum \eqref{eq:id_i-cdot-id_j} we have $q'=q$, then clearly $\vec{R}_q \cdot \vec{R}_q \neq 0$, since $R$ is a row-stochastic matrix. Thus, since all of the terms in the sum are nonnegative, in order for it to result zero, we must have that $(LC)_{iq}  (LC)_{jq}= 0$ for all $q \in \{1, \ldots, m\}$ and $i \neq j$. Hence,
\begin{align}
\nonumber 0 &= \overrightarrow{LC}_i \cdot \overrightarrow{LC}_j = \sum_{q = 1}^m (LC)_{iq} (LC)_{jq} = \sum_{q=1}^m \sum_{p,p'=1}^n L_{ip}C_{pq} L_{jp'}C_{p'q} \\
&= \sum_{p,p'=1}^n \left( \sum_{q=1}^m C_{pq}C_{p'q} \right) L_{ip} L_{jp'} = \sum_{p,p'=1}^n \left( \vec{C}_p \cdot \vec{C}_{p'} \right) L_{ip} L_{jp'} \, . \label{eq:LC_i-cdot-LC_j}
\end{align}

Clearly, since $L$ is a row-stochastic matrix, for all $i\neq j$  there exists $p_i, p_j \in \{1, \ldots, n\}$ such that $L_{ip_i} L_{j p_j} \neq 0$. Then we again see that in order for the  sum \eqref{eq:LC_i-cdot-LC_j} (that again consists of nonnegative terms) to result zero, we must have  $\vec{C}_{p_i} \cdot \vec{C}_{p_j} = 0$, i.e., the $p_i$th and $p_j$th rows of $C$ must be orthogonal. Thus, there exist at least $k$ indices $\{p_1, \ldots, p_k\} \subset \{1, \ldots, n\}$ such that $\vec{C}_{p_i} \cdot \vec{C}_{p_j}=0$ for all $i \neq j$. In order to complete this part of the proof we need to show that $p_i \neq p_j$ for all $i\neq j$. 

For that, let us consider the sum \eqref{eq:LC_i-cdot-LC_j} once more. For $p'=p$ in the sum, we obviously have that $\vec{C}_{p} \cdot \vec{C}_{p'} =\vec{C}_{p} \cdot \vec{C}_{p} \neq 0$, so that we must have $L_{ip} L_{jp}=0$ for all $p \in \{1, \ldots, n\}$ and $i \neq j$. Thus, if we would have $p_i=p_j$ for some $i \neq j$, then $L_{ip_i}L_{jp_j} = 0$ which contradicts the defining property $L_{ip_i}L_{jp_j} \neq 0$ of how $p_i$ and $p_j$ were found. We conclude that $p_i \neq p_j$ for all $i \neq j$ so that the rows $\vec{C}_{p_1}, \ldots, \vec{C}_{p_k}$ are orthogonal.

Let now $\iota(C)=k$. Thus, by the previous part of the proof $C$ has (at least) $k$ orthogonal rows. Suppose that $C$ has $s>k$ orthogonal rows. Then by the first part of the proof we would have that $L'CR' = \id_s$ for some stochastic matrices $L'$ and $R'$ so that $\iota(C) \geq s > k = \iota(C)$, which is a contradiction.

On the other hand, let the maximal number of orthogonal rows of $C$ be $k$ and suppose that $t:=\iota(C) > k$.  Thus, there exist stochastic $L$ and $R$ such that $LCR=\id_t$ and again by the first part of the proof we would have that $C$ has at least $t$ orthogonal rows which contradicts the maximality of $k$. This concludes the proof.
\end{proof}

\begin{example}
The proof of the previous proposition gives us a technique to find $L, R \in \Mrall$ for $C \in \Mrall$ such that $LCR = \id_{\iota(C)}$. 
However, we note that the $L$ and $R$ are however not unique; let us take 
\begin{equation}
C= \begin{bmatrix}[1.2]
\frac{1}{2} & \frac{1}{2} & 0 & 0 & 0 \\
\frac{1}{3} & 0 & \frac{1}{3} & \frac{1}{3} & 0 \\
0 & \frac{1}{2} & 0 & 0 & \frac{1}{2} \\
0 & 0 & 0 & \frac{1}{2} & \frac{1}{2}  \\
\end{bmatrix}.
\end{equation}
We see that the maximal number of orthogonal rows is 2 and this is the case for rows $(p_1,p_2) = (1,4)$ and $(p'_1,p'_2) = (2,3)$. Let first $(p_1,p_2)=(1,4)$. By following the proof, we arrive at the following majorization decomposition for $\id_2$:
\begin{equation*}
\begin{bmatrix}
1 & 0 & 0 & 0 \\
0 & 0 & 0 & 1 \\
\end{bmatrix}
\begin{bmatrix}[1.2]
\frac{1}{2} & \frac{1}{2} & 0 & 0 & 0 \\
\frac{1}{3} & 0 & \frac{1}{3} & \frac{1}{3} & 0 \\
0 & \frac{1}{2} & 0 & 0 & \frac{1}{2} \\
0 & 0 & 0 & \frac{1}{2} & \frac{1}{2}  \\
\end{bmatrix}
\begin{bmatrix}[1.05]
1 & 0 \\
1 & 0 \\
\frac{1}{2} & \frac{1}{2} \\
0 & 1 \\
0 & 1 \\
\end{bmatrix}
 = \begin{bmatrix}
 1 & 0 \\ 
 0 & 1 \\
 \end{bmatrix}
\end{equation*}
Similarly, if we take $(p'_1,p'_2) =(2,3)$, we get the following decomposition:
\begin{equation*}
\begin{bmatrix}
0 & 1 & 0 & 0 \\
0 & 0 & 1 & 0 \\
\end{bmatrix}
\begin{bmatrix}
\frac{1}{2} & \frac{1}{2} & 0 & 0 & 0 \\
\frac{1}{3} & 0 & \frac{1}{3} & \frac{1}{3} & 0 \\
0 & \frac{1}{2} & 0 & 0 & \frac{1}{2} \\
0 & 0 & 0 & \frac{1}{2} & \frac{1}{2}  \\
\end{bmatrix}
\begin{bmatrix}
1 & 0 \\
0 & 1 \\
1 & 0 \\
1 & 0 \\
0 & 1 \\
\end{bmatrix}
 = \begin{bmatrix}
 1 & 0 \\ 
 0 & 1 \\
 \end{bmatrix}.
\end{equation*}
\end{example}

\section{Example: uniform (anti)distinguishability}\label{sec:example}

As a demonstration of the ultraweak order and the introduced tools, we analyze tha family of communication matrices $D_{n,\epsilon}$, where $n$ is fixed and $0\leq\epsilon\leq 1$.
As discussed in Section \ref{sec:some}, a communication matrix $D_{n,\epsilon}$ is related either to a distinguishability task or to an antidistinguishability task, the boundary value being $\epsilon=1-1/n$.
The special cases are $D_{n,0}=\id_n$ (error-free distinguishability), $D_{n,1} \simeq A_n$ (error-free antidistinguishability) and $D_{n,1-1/n}=V_n$ (pure noise).
It is clear that $V_n \uleq D_{n,\epsilon} \uleq \id_n$ for any $\epsilon$.

\begin{figure}
\begin{center}
\includegraphics[scale=1]{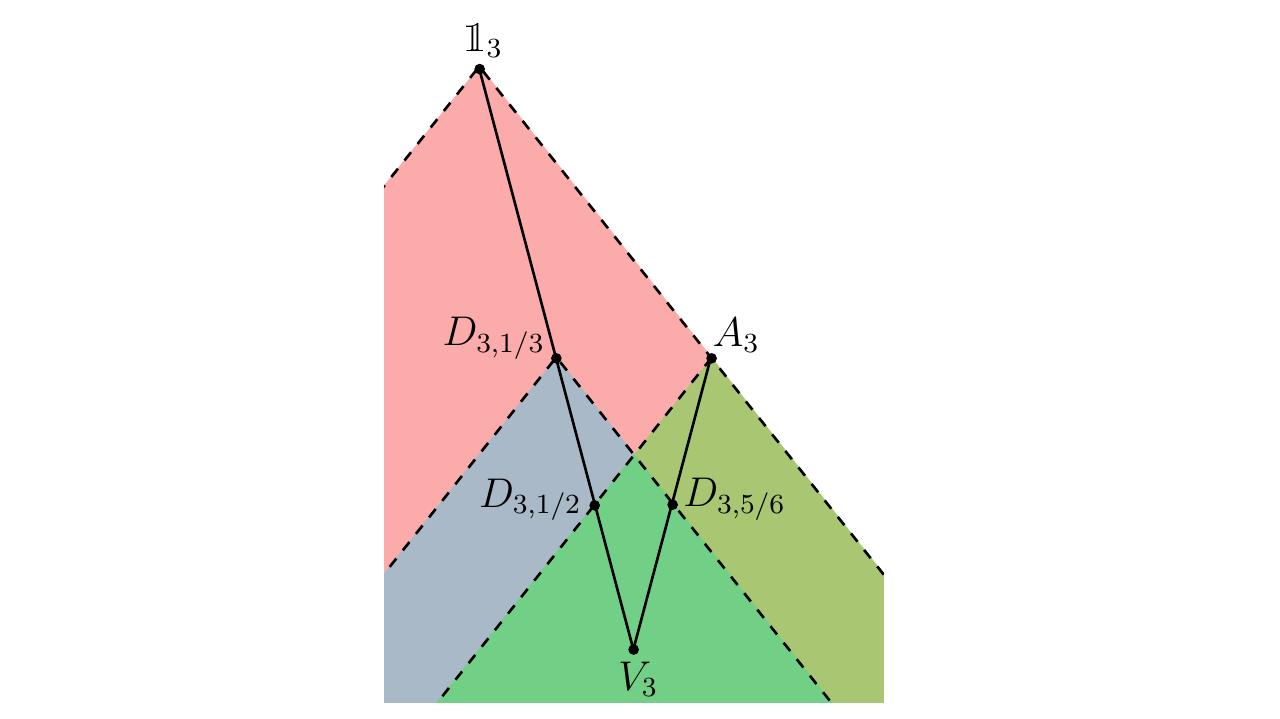}
\caption{\label{fig-2} The ultraweak order of matrices $D_{n,\epsilon}$ for $n=3$. The line segment is formed of the matrices $D_{3,\epsilon}$, and a matrix is majorized by another matrix if the first matrix is contained in the downward cone placed at the position of the second matrix.}
\end{center}
\end{figure}

In the following we determine the ultraweak order completely for fixed $n$.
The results are:
\begin{itemize}
\item  If $\epsilon\in[0,1-1/n]$ (i.e. distinguishability), then $D_{n,\epsilon} \ugeq D_{n, \mu}$ if and only if $\mu \in \left[ \epsilon, 1- \frac{\epsilon}{n-1} \right]$.
\item If $\epsilon\in[1-1/n,1]$ (i.e. antidistinguishability), then $D_{n,\epsilon} \ugeq D_{n, \mu}$ if and only if $\mu \in \left[1- \frac{\epsilon}{n-1}, \epsilon\right]$. 
\end{itemize}
Interestingly, the error-free antidistinguishability matrix $D_{n,1}$ is majorized only by $\id_n$. The results for $n=3$ are illustrated in Fig. \ref{fig-2}. The black line segments are formed by the matrices $D_{3,\epsilon}$ for all $\epsilon \in [0,1]$ such that $\id_3$ (with $\epsilon=0$) is on top and as the line segment goes down to $V_3$, the corresponding value of $\epsilon$ goes up to $2/3$, and as the line segment climbs back up from $V_3$ to $A_3$, the values of $\epsilon$ grow from $2/3$ to $1$. A matrix $D_{3, \mu}$ can be majorized by a matrix $D_{3,\epsilon}$ if and only if it is contained in the downward cone placed at $D_{3,\epsilon}$. We see for example that every $D_{3,\epsilon}$ can be majorized from the identity $\id_3$, and $D_{3,\epsilon}$ can be majorized by $A_3$ if and only if $\epsilon\in [1/2, 1]$. The space outside the line segments reperesent other communication matrices so that for example everything inside the red downward cone starting from $\id_3$ is exactly the set of all communication matrices obtained from a $3$-dimensional classical theory. The cones are cut for practical purposes.

We start by showing that the bounds for $\mu$ for a given $\epsilon$ are necessary for the majorization $D_{n,\epsilon} \ugeq D_{n, \mu}$. For this we will use the previously defined monotones $\lambda_{max}$ and $\lambda_{min}$. 

First, let $\epsilon\in[0,1-1/n]$ and $D_{n, \epsilon} \ugeq D_{n, \mu}$. Suppose that $ 0 \leq \mu < \epsilon$. Then by using $\lambda_{max}$, we see that
\begin{equation*}
\lambda_{max}(D_{n, \epsilon}) = n(1- \epsilon) \geq n(1-\mu) = \lambda_{max}(D_{n, \mu}) \, ,
\end{equation*}
which implies that $\mu \geq \epsilon$ thus leading to a contradiction. Similarly, if we suppose that $\mu > 1- \frac{\epsilon}{n-1}$, by using $\lambda_{min}$, we see that
\begin{equation*}
\lambda_{min}(D_{n, \epsilon}) = -\frac{n}{n-1}\epsilon \geq -n(1-\mu) = \lambda_{min}(D_{n, \mu}) \, ,
\end{equation*}
which then implies that $\mu \leq 1- \frac{\epsilon}{n-1}$ thus again leading to a contradiction. We conclude that if $\epsilon \in [0,1-1/n]$, then $\mu \in \left[ \epsilon, 1- \frac{\epsilon}{n-1} \right]$.

Second, let $\epsilon \in [1-1/n, 1]$ and $D_{n, \epsilon} \ugeq D_{n, \mu}$. 
By a similar fashion, if we suppose that $0 \leq \mu < 1-\frac{\epsilon}{n-1}$, by using $\lambda_{max}$ we see that 
\begin{equation*}
\lambda_{max}(D_{n, \epsilon}) = \frac{n}{n-1}\epsilon \geq n(1-\mu) = \lambda_{max}(D_{n, \mu}) \, ,
\end{equation*}
which implies that $\mu \geq 1- \frac{\epsilon}{n-1}$ which is again a contradiction. On the other hand if we suppose that $\mu > \epsilon$, then
\begin{equation*}
\lambda_{min}(D_{n, \epsilon}) = -n(1-\epsilon) \geq -n(1-\mu) = \lambda_{min}(D_{n, \mu}) \, ,
\end{equation*}
and thus $\mu \leq \epsilon$. Hence, we conclude that if $\epsilon\in[1-1/n,1]$, then $\mu \in \left[1- \frac{\epsilon}{n-1}, \epsilon\right]$.

To see that we can actually realize all the stated majorizations,  we define a row-stochastic matrix $L_\lambda =\lambda \id_n + (1-\lambda) D_{n,1}$ for any $\lambda \in [0,1]$.
Then
\begin{align*}
L_\lambda D_{n,\epsilon} & = \lambda D_{n,\epsilon} + (1-\lambda) D_{n,1-\frac{\epsilon}{n-1}} = D_{n, 1-\lambda (1- \epsilon) -(1-\lambda) \frac{\epsilon}{n-1}} \, .
\end{align*}
Now if  $\epsilon\in[0,1-1/n]$, then $\epsilon \leq 1- \frac{\epsilon}{n-1}$ so that  $1-\lambda (1- \epsilon) -(1-\lambda) \frac{\epsilon}{n-1} \in \left[ \epsilon, 1- \frac{\epsilon}{n-1} \right]$ for all $\lambda \in [0,1]$ and similarly if $\epsilon\in[1-1/n,1]$, then $\epsilon \geq 1- \frac{\epsilon}{n-1}$ so that  $1-\lambda (1- \epsilon) -(1-\lambda) \frac{\epsilon}{n-1} \in \left[ 1- \frac{\epsilon}{n-1}, \epsilon \right]$ for all $\lambda \in [0,1]$.

Let us then consider a state space $\state$ with operational dimension $d_{op}$ and linear dimension $d_{lin}$.
We denote 
\begin{equation*}
I_{n}(\state) = \{ \epsilon \in [0,1] \,|\, D_{n,\epsilon}\in\Call{\state} \} \, .
\end{equation*}
The previous characterization of the ultraweak majorization of the $D_{n,\epsilon}$ matrices implies that $I_{n}(\state)$ is an interval. To further limit the interval $I_n(\state)$, we see that if $\epsilon \in I_n(\state)$ so that $D_{n,\epsilon} \in \Call{\state}$, then $n(1-\epsilon) = \lmax{D_{n,\epsilon}} \leq \lambda_{max,n}(\state)$ which gives us $\epsilon \geq 1- \lambda_{max,n}(\state)/n$. This is precisely the generalization of the basic decoding theorem that was discussed in Sec. \ref{subsec:minmax}. Thus, we must always have that 
\begin{equation}\label{eq:I_n-lmax}
I_n(\state) \subseteq \left[ 1- \frac{\lambda_{max,n}(\state)}{n}, 1 \right]
\end{equation}

Firstly, for $n\leq d_{op}$ we have $D_{n,\epsilon}\in\Call{\state}$ for all $\epsilon\in[0,1]$ as in these cases $D_{n,\epsilon}\uleq\id_n\in\Call{\state}$. In this case also $\lambda_{max,n}(\state) = n$ so that Eq. \eqref{eq:I_n-lmax} gives only a trivial lower bound.
Secondly, for $n\geq d_{lin}+1$ we have $D_{n,\epsilon}\in\Call{\state}$ only when $\epsilon= 1-1/n$ since $\rank{D_{n,\epsilon}}=n$ for $\epsilon\neq1-1/n$ and Prop. \ref{prop:maximal-rank} bounds the rank of matrices in $\Call{\state}$.
The nontrivial cases are hence limited to $d_{op}<n\leq d_{lin}$ and the intervals $I_{n}(\state)$ then depend on the details of $\state$.

\begin{example}(\emph{Uniform (anti)distinguishability of qudit states})
As was mentioned above, the nontrivial cases of determining $I_n(\qudit)$ are limited to $d<n\leq d^2$ since $I_n(\qudit) = [0,1]$ for $n\leq d$ and $I_n(\qudit) = \{1 - \frac 1n \}$ for $n > d^2$. For $d<n\leq d^2$, we get from the basic quantum decoding theorem  (see Example \ref{ex:lmax-quantum}) that $I_n(\qudit) \subseteq \left[ 1- \frac dn, 1 \right]$. Furthermore, for $n<d^2$, we see from Example  \ref{ex:G42-psdrank} that $\psdrank{A_n} \leq d$ so that $1 \in I_n(\qudit)$ and thus the upper bound in the previous inclusion is attained in this case. By Example \ref{ex:G42-psdrank} this also holds for all $n=d^2$ when $d\leq 151$ or $d$ is odd. The lower bound is also attained for $n=d^2$ when $d \leq 151$ as it is then given by a SIC-POVM.

In the special case of qubits, one can show that for $n=3$ and $n=4$ we can choose pure qubit states with equal trace distance and perform their minimum error discrimination that gives $D_{n, 1 - \frac 2n}$. Thus, for qubits we conclude that  $I_{n}(\qubit) = \left[ 1- \frac 2n, 1 \right]$ for the nontrivial cases when $n\in \{3,4\}$.
\end{example}

\section{Comparisons of the monotones}\label{sec:comparison}

In this section we examine the monotones further. 
In Subsec. \ref{sec:mono1} we show that all of the introduced monotones are useful, meaning that each of them detect some inequivalences that none other is detecting.  
In Subsec. \ref{sec:mono2} we show that our collection of six monotones is not enough to characterize the ultraweak matrix majorization. 
This may not be a surprise, but we believe that analyzing inequivalent matrices where all the monotones coincide may give insight in introducing other monotones. 
Finally, in Subsec. \ref{sec:mono3} we summarize some inequalities for the numerical values of the monotones.

\subsection{Incomparability of the monotones}\label{sec:mono1}

Let $f_1$ and $f_2$ be two ultraweak monotones.
We say that $f_1$ is \emph{finer} than $f_2$ if the following holds for all $C,D\in\Mrall$:
\begin{align*}
f_1(C)=f_1(D) \quad \Rightarrow \quad f_2(C) = f_2(D) \, .
\end{align*}
This relation is a preorder in the set of ultraweak monotones. 
Obviously, if we have two monotones and one of them is finer than the other, then we can ignore the second one and use only the finer one.

In the following we demonstrate that all the introduced ultraweak monotones are incomparable in the sense that none is finer than any other.
We will present a set of matrices where, for a given monotone, there is a pair of communication matrices such that their ultraweak inequivalence is only detected by the given monotone and all the other monotones take the same value for both matrices. As we give such an example for each monotone, together they show that each of them can be used to detect an inequivalence that the others cannot detect.
This means that they are also incomparable in the sense described earlier. 

Let us consider the following matrices:
\begin{align*}
K_+ &= \frac{1}{2} \begin{bmatrix}
1 & 1 & 0 & 0 \\
1 & 0 & 1 & 0 \\
1 & 0 & 0 & 1 \\
0 & 1 & 0 & 1 \\
0 & 0 & 1 & 1 \\
\end{bmatrix}, \quad 
K= \frac{1}{2} \begin{bmatrix}
1 & 1 & 0 & 0 \\
1 & 0 & 1 & 0 \\
0 & 1 & 0 & 1 \\
0 & 0 & 1 & 1 \\
\end{bmatrix},  \quad
K_- = \frac{1}{2} \begin{bmatrix}
1 & 1 & 0 & 0 \\
1 & 0 & 1 & 0 \\
0 & 1 & 0 & 1 \\
\end{bmatrix},  \\
D_{3,1/3} &=  \begin{bmatrix}[1.3]
\frac{2}{3} & \frac{1}{6} & \frac{1}{6} \\
\frac{1}{6} & \frac{2}{3} & \frac{1}{6} \\
\frac{1}{6} & \frac{1}{6} & \frac{2}{3} \\
\end{bmatrix}, \quad  \quad \quad 
A= \begin{bmatrix}[1.3]
1 & 0 & 0 \\
\frac{1}{2} & \frac{1}{2} & 0 \\
\frac{1}{2} & 0 & \frac{1}{2} \\
\end{bmatrix}, \quad  \quad \quad 
B= \begin{bmatrix}[1.3]
\frac{2}{3} & \frac{1}{3} & 0 \\
0 & \frac{2}{3} & \frac{1}{3} \\
\frac{1}{3} & 0 & \frac{2}{3} \\
\end{bmatrix}, \\
C &= \begin{bmatrix}[1.3]
1 & 0 & 0 \\
0& \frac{1}{2} & \frac{1}{2}  \\
\frac{1}{2} & 0 & \frac{1}{2} \\
\end{bmatrix},\quad  \quad \quad 
D= \begin{bmatrix}[1.3]
1 & 0 & 0 \\
0& \frac{1}{2} & \frac{1}{2}  \\
0 & 0 & 1 \\
\end{bmatrix}.
\end{align*}
The values of the monotones for the matrices above are listed in Table \ref{table:monotones}. 

\begin{table}[h]
\begin{tabular}{|l|c|c|c|c|c|c|c|c|c|} \hline
					& $K_+$ & $K$ & $K_-$ & $D_{3,1/3}$ & $A$ & $B$ & $C$ & $D$  \\ \hline
$\rankw$		& \blue{4} & \blue{3} & 3 & 3    & 3    & 3 & 3 & 3   \\ \hline
$\nrankw$		& 4 & \blue{4} & \blue{3} & 3    & 3    & 3 & 3 & 3   \\ \hline
$\psdrankw$	& 3 & 3 & 3 & \blue{2}    & \blue{3}    & 3 & 3 & 3   \\ \hline
$\lminw$		& 0 & 0 & 0 & -1/2 & \blue{-1/2} & \blue{0} & 0 & 0   \\ \hline
$\iota$ 		& 2 & 2 & 2 & 1    & 1    & \blue{1} & \blue{2} & 2   \\ \hline
$\lmaxx$		& 2 & 2 & 2 & 2    & 2    & 2 & \blue{2} & \blue{5/2} \\ \hline
\end{tabular} \vspace*{0.3cm}
\caption{ \label{table:monotones} The values of the monotones for the given matrices. For each monotone one can find a pair of matrices for which only that monotone can be used to detect the inequivalence of the pair.}
\end{table}

We see that for each monotone, there is a pair of matrices such that only the monotone in question detects the ultraweak inequivalence of the matrices. The monotones and the pairs of matrices are the following: rank for $K_+$ and $K$, nonnegative rank for $K$ and $K_-$, positive semidefinite rank for $D_{3,1/3}$ and $A$, $\lambda_{min}$ for $A$ and $B$, $\iota$ for $B$ and $C$, and lastly $\lambda_{max}$ for $C$ and $D$. Together they show that the monotones presented here are incomparable in the sense described earlier.

Next we explain how the values of the monotones in Table \ref{table:monotones} were obtained. The values for $\mathrm{rank}$, $\lambda_{max}$ and $\lambda_{min}$ can easily be obtained from their definitions and require no further analysis. Also the characterization in Prop. \ref{prop:iota} for $\iota$ gives an easy way to calculate the value of $\iota$. Thus, what remain are the values of $\mathrm{rank}_+$ and $\mathrm{rank}_{psd}$ for the given matrices.

First we note that if for a matrix $M \in \Mr{n}{m}$ we have that $\rank{M} = \min\{n,m\}$, then by Eq. \eqref{eq:rank-ineq} we must have that $\nrank{M} = \min\{n,m\}$. Thus, by using this we get the values of $\mathrm{rank}_+$ for every other matrix other than $K$. 
The nonnegative rank of $K$ was found in Ref. \cite{CoRo93} to be $4$.

For the positive semidefinite rank, we start with the matrices $K_+$, $K$ and $K_-$. Since $K_+$ can be obtained from $G_{4,2}$ by removing one row and similarly $K$ from $K_+$ and $K_-$ from $K$, we have the following chain of majorizations:
\begin{equation}
K_- \uleq K \uleq K_+ \uleq G_{4,2} \, .
\end{equation}
Since we know from Example \ref{ex:G42-psdrank} that $\psdrank{G_{4,2}} =3$, we thus get that the positive semidefinite rank of $K_+$, $K$ and $K_-$ must all be less than three. However, all three matrices satisfy the premises of Prop. \ref{prop:psdrank-kernel} so that together with the previous fact we must have that the positive semidefinite rank of all three matrices is exactly three. 

For the remaining matrices, we note that from Prop. \ref{prop:psdrank-kernel} and Eq. \eqref{eq:psd-rank-ineq} it follows that $\psdrank{A} = \psdrank{C} = \psdrank{D} = 3$. 
From Example \ref{ex:nonconvex} we conclude that $\psdrank{D_{3,1/3}} =2$ and $\psdrank{B} \geq 3$. 
However, from Eq. \eqref{eq:psd-rank-ineq} we see that $\psdrank{B} = 3$. 
This concludes the filling of the table.

\subsection{Incompleteness of the set of monotones}\label{sec:mono2}

From Table \ref{table:monotones} we see that all the monotones coincide for the matrices $C$ and $K_-$. 
It is straightforward to see that $C \uleq K_-$. We will continue to show that even though the values of the monotones coincide for $C$ and $K_-$, they are not ultraweakly equivalent so that our set of monotones is not a complete set. This means that our monotones do not always detect the inequivalence of some communication matrices.

\begin{proposition}
$K_- \npreceq C$.
\end{proposition}
\begin{proof}
Suppose that $K_- \uleq C$ so that there exists $L \in \Mr{3}{3}$ and $R \in  \Mr{3}{4}$ such that $LCR = K_-$. By calculating $LCR$ explicitly we find that \begin{align}\label{eq:LCR-elements}
(LCR)_{ij} = \dfrac{1}{2} \left[ (2 L_{i1} + L_{i3}) R_{1j} + L_{i2} R_{2j} + (L_{i2} +L_{i3})R_{3j} \right]
\end{align}
and this must coincide with $(K_-)_{ij}$ for all $i\in \{1,2,3\}$ and $j \in \{1,2,3,4\}$. First we note that since $K_-$ does not have any zero columns, we must have that all the columns of $R$ must also be nonzero so that each column of $R$ must have at least one nonzero element. 

Let us focus on the third column of $R$. From the above consideration we know that at least one of the elements $R_{13}$, $R_{23}$ or $R_{33}$ must be nonzero in order to $LCR =K_-$ to hold. Suppose first that $R_{13} \neq 0$. By looking at the third column of $K_-$, we find that 
$(K_-)_{13} = (K_-)_{33}=0$ so that from Eq. \eqref{eq:LCR-elements} it follows that $L_{i1} = L_{i3}=0$ and $L_{i2}=1$ for $i \in \{1,3\}$. Then by plugging these into the expressions of the elements $(LCR)_{11}$ and $(LCR)_{31}$ in Eq. \eqref{eq:LCR-elements}, we find that 
\begin{align*}
0 = (K_-)_{31} = (LCR)_{31} = \dfrac{1}{2} (R_{21} + R_{31}) = (LCR)_{11} = (K_-)_{11} = \frac{1}{2} \, ,
\end{align*}
which is a contradiction. 
Thus, $R_{13}=0$. 

Suppose next that $R_{33} \neq 0$ so that again from the third column of $K_-$ we can deduce with Eq. \eqref{eq:LCR-elements} that $L_{i2} = L_{i3} = 0$ and $L_{i1}=1$ for $i \in \{1,3\}$. By plugging these into the expressions of the elements $(LCR)_{11}$ and $(LCR)_{31}$ we find that 
\begin{align*}
0 = (K_-)_{31} = (LCR)_{31} = R_{11} = (LCR)_{11} = (K_-)_{11} = \frac{1}{2} \, ,
\end{align*}
which is again a contradiction. 
Thus, also $R_{33}=0$.

Suppose lastly that $R_{23} \neq 0$. Once again we use the third column of $K_-$ to find that in this case we must have $L_{12} = L_{32} =0$. If we look at elements $(LCR)_{11}$ and $(LCR)_{31}$ by noting that 
\begin{equation*}
2 L_{i1} + L_{i3} = 1+ L_{i1}
\end{equation*}
 for $i\in \{1,3\}$ which follows from the stochasticity of $L$, we find that 
\begin{align*}
\frac{1}{2} [(1+ L_{11})R_{11} + L_{13} R_{31}] = (LCR)_{11} = (K_-)_{11} = \frac{1}{2} \, , \\
\frac{1}{2} [(1+ L_{31})R_{11} + L_{33} R_{31}] = (LCR)_{31} = (K_-)_{31} = 0 \, .
\end{align*}
From the latter expression we can deduce that $R_{11}=0$ so that for the former expression it means that $L_{13} = R_{31} = 1$. For the latter expression this means also that $L_{33}=0$ which together with $L_{32}=0$ and stochasticity of $K$ implies that $L_{31}=1$. From the stochasticity of $L$ it also follows that also $L_{11}=0$. By plugging these values into the expressions of the elements $(LCR)_{14}$ and $(LCR)_{34}$ we see that 
\begin{align*}
\frac{1}{2} [R_{14} + R_{34}] = (LCR)_{14} = (K_-)_{14} = 0, \\
R_{14} = (LCR)_{34} = (K_-)_{34} = \dfrac{1}{2},
\end{align*}
which together imply that $\frac{1}{2}= R_{14} = 0$, which is a contradiction. Hence, also $R_{23}=0$, but since none of the columns of $R$ could not be zero in order to satisfy $LCR=K_-$, we are forced to conclude that there does not exist such $L$ and $R$. 
\end{proof}

\subsection{Comparison of numerical values of the monotones}\label{sec:mono3}

Let us consider the monotones $\iota$,  $\lambda_{max}$, $\mathrm{rank}_{psd}$, $\mathrm{rank}$ and $\mathrm{rank}_+$. By minimizing and maximizing over all communication matrices of an operational theory with a state space $\state$, we have the following bounds for these monotones for all $C \in \Call{\state}$:
\begin{align*}
1 &\leq \iota(C) 	\leq d_{op}(\state) \\
1 &\leq \lmax{C} 	\leq \lmax{\state} \\
1 &\leq \psdrank{C} 	\leq d_q(\state) \\
1 &\leq \rank{C}		\leq d_{lin}(\state) \\
1 &\leq \nrank{C} 	\leq d_{cl}(\state)
\end{align*}

The numerical values of the monotones are related in the following way.

\begin{proposition}
For all $C \in \Call{\state}$, we have that
\begin{equation}\label{eq:comp1}
\iota(C) \leq \lmax{C} \leq \psdrank{C} \leq \nrank{C} 
\end{equation}
and
\begin{equation}\label{eq:comp2}
\sqrt{\rank{C}} \leq \psdrank{C}, \quad \rank{C} \leq \nrank{C} \, .
\end{equation}
\end{proposition}

\begin{proof}
The first inequality in Eq. \eqref{eq:comp1} follows from noting that $\id_{\iota(C)} \uleq C$ so that by the monotonicity of $\lambda_{max}$ we must have $\lmax{C} \geq \lmax{\id_{\iota(C)}} = \iota(C)$. The second inequality in Eq. \eqref{eq:comp1} is the same as Eq. \eqref{eq:psdrank-lmax} which was shown in \cite{LeWeWo17} and can be obtained from Example \ref{ex:lmax-quantum} and Prop. \ref{prop:psdrank-quantum}. The final inequality in Eq. \eqref{eq:comp1} is just the second inequality in Eq. \eqref{eq:psd-rank-ineq}, and the inequalities in Eq. \eqref{eq:comp2} are collected from Eq. \eqref{eq:psd-rank-ineq} and Eq. \eqref{eq:rank-ineq}.
\end{proof}

By taking the supremum of each monotone in a given theory with a state space $\state$, we get the following corollary.
\begin{corollary}
For all state spaces $\state$ we have that
\begin{equation}\label{eq:cor-comp1}
d_{op}(\state) \leq \lmax{\state} \leq d_q(\state) \leq d_{cl}(\state)
\end{equation}
and
\begin{equation}\label{eq:cor-comp2}
d_{lin}(\state) \leq d_q(\state)^2, \quad d_{lin}(\state) \leq d_{cl}(\state).
\end{equation}
\end{corollary}

The dimensions given by the monotones can be used to categorize theories into different classes where in each class each theory has the same set of dimensions. One can then compare communication in theories of different classes together as the dimensions are indicators to what kind of communication tasks might be possible to implement. To find differences in theories within the same class with the same dimensions one would have to consider the set of communication matrices themselves and see if some task can be implemented in some theory and not in others. Lastly, we demonstrate the dimensions in classical and quantum theory.

For the classical state space $\state^{cl}_d$ we have that all of the previous dimensions give the same value $d$, i.e.
\begin{align*}
d_{op}(\state^{cl}_d) = \lmax{\state^{cl}_d} = d_q(\state^{cl}_d) = d_{lin}(\state^{cl}_d) = d_{cl}(\state^{cl}_d) =d,
\end{align*}
and one can confirm that the property that all of the dimensions are equal is also necessary for a state space to be classical (actually it suffices only for $d_{op}$ and $d_{lin}$ to be equal). Thus, $d$-dimensional classical theory forms a class of its own. We note that classical theory saturates all the inequalities in Eq. \eqref{eq:cor-comp1} and the latter inequality in Eq. \eqref{eq:cor-comp2}. 

For quantum state space $\qudit$, we have that
\begin{align*}
d_{op}(\qudit) = \lmax{\qudit} = d_q(\qudit) = d, \quad d_{lin}(\qudit) = d^2 \leq d_{cl}(\qudit).
\end{align*}
Thus, for $d$-dimensional quantum theory, the first three inequalities in Eq. \eqref{eq:cor-comp1} are saturated as well as the first inequality of Eq. \eqref{eq:cor-comp2}. Our conjecture is that also the second inequality of Eq. \eqref{eq:cor-comp2} is saturated for $\qudit$ so that $d_{cl}(\qudit) = d^2$. However, to the best of our knowledge, $d_{cl}(\qudit)$ is not known, and it remains an open problem for future research. An open problem is also the question which saturations of inequalities, if any, in Eqs. \eqref{eq:cor-comp1} and \eqref{eq:cor-comp2} is enough to for a theory to be fixed to be $\qudit$, or if there are other theories with the same set of dimensions as $\qudit$.

\section{Discussion}
We have considered communication matrices that can be seen as communication tasks in prepare-and-measure scenarios in operational theories. We have introduced a preorder, ultraweak matrix majorization, in the set of communication matrices that captures the idea of when a communication task is harder than some other. We have further developed a set of monotones for the preorder, examined their properties and seen that they are can give information about the physical properties of a given theory. We demonstrate our results with concrete communication matrices that come from motivated physical communication tasks.

Our first monotone, the rank, is linked to the linear dimension of a state space, and it can be thus used as a necessary mathematical criterion to see if a communication task can be implemented with a theory with a specific dimension. Our second (third) monotone, the nonnegative (positive semidefinite) rank, gives the minimal dimension of a classical (quantum) system needed to implement a given communication task. It has been shown that computing the nonnegative rank and the (real) positive semidefinite rank of a matrix are NP-hard problems \cite{Vavasis09, Shitov16} in general. We do however provide a lower bound for the positive semidefinite rank of some special types of matrices. 

Our fourth and fifth monotes, the max and the min monotone, calculates the sum of the maximal and minimal elements of each column. The max monotone is then seen to link to some decoding and minimum error discrimination tasks. Our last monotone, the distinguishability monote, tells us how many (error-free) messages can be sent with a given communication matrix, and it is seen to connect to the operational dimension of a theory. We show that all the defined monotones are useful for detecting an inequivalence of different tasks but that they do not characterize the preorder completely. 
Finally, we consider and compare certain type of dimensions given by the monotones and discuss how they can be used to classify different operational theories.

An intriguing open problem is to characterize the set $\Call{\qudit}$ of all qudit implementable communication matrices. 
In the presented framework this amounts to finding the maximal elements of $\Call{\qudit}$ in the ultraweak preordering.
A similar question can naturally be posed for any operational theory $\state$, and this then gives a physically motivated way to compare different theories.

Another point to make is that we have been considering a state space $\state$ and taking all mathematically possible measurements as physical measurements. 
This corresponds to the so-called no-restriction hypothesis. 
One can relax this hypothesis and have a milder assumption that the set of measurements is simulation closed \cite{FiGuHeLe19}. 
This is enough to guarantee that the set of implementable communication matrices is closed in the ultraweak majorization and the introduced mathematical machinery hence makes sense. 

\section{Acknowledgements}

We wish to thank Julio I de Vicente for his helpful comments related to Example \ref{ex:G42-psdrank}. This work was performed as part of the Academy of Finland Centre of Excellence program, Project 312058. Financial support from the Academy of Finland, Project 287750, is also acknowledged. O.K. would like to acknowledge the financial support from the Turku University Foundation. L.L. acknowledges financial support from University of Turku Graduate School (UTUGS).

\end{document}